\documentclass[lettersize,journal]{IEEEtran}
\IEEEoverridecommandlockouts

\usepackage{amsmath,amssymb,amsfonts}
\usepackage{graphicx}
\usepackage{textcomp}
\usepackage{url}
\usepackage{xcolor}
\usepackage{todonotes}
\usepackage{dsfont}
\usepackage[square,numbers]{natbib}
\usepackage{hyperref}
\usepackage{easyReview}
\usepackage{quantikz}
\usepackage{algorithm}
\usepackage{algpseudocode}
\usepackage{mathtools}
\usepackage{longtable}
\usepackage{longtable}
\usepackage{enumitem}
\usepackage{balance} 
\usepackage{orcidlink}
\usepackage{listings}
\usepackage{tabularx}
\usepackage{booktabs}
\usepackage{array}
\usepackage{multirow}
\usepackage[most]{tcolorbox}

\usepackage{xcolor}
\usepackage{fvextra} 
\usepackage{subcaption}
\usepackage{caption}
\usepackage{float}

\DefineVerbatimEnvironment{CodeBlock}{Verbatim}
{formatcom=\color{blue}, fontsize=\small}

\usepackage{tikz}
\usetikzlibrary{positioning,arrows.meta,fit,backgrounds}

\newcommand{\Adv}{\mathsf{Adv}}

\newcommand{\LocGen}{\mathsf{LocGen}}
\newcommand{\Emb}{\mathsf{Emb}}
\newcommand{\Comp}{\mathsf{Comp}}
\newcommand{\Ext}{\mathsf{Ext}}
\newcommand{\Read}{\mathsf{Read}}
\newcommand{\Write}{\mathsf{Write}}

\definecolor{codegreen}{rgb}{0,0.6,0}
\definecolor{codegray}{rgb}{0.5,0.5,0.5}
\definecolor{codepurple}{rgb}{0.58,0,0.82}
\definecolor{backcolour}{rgb}{0.95,0.95,0.92}
\lstdefinestyle{pythonstyle}{
	backgroundcolor=\color{backcolour},
	commentstyle=\color{codegreen},
	keywordstyle=\color{magenta},
	numberstyle=\tiny\color{codegray},
	stringstyle=\color{codepurple},
	basicstyle=\ttfamily\tiny,
	breakatwhitespace=false,
	breaklines=true,
	captionpos=b,
	keepspaces=true,
	numbers=left,
	numbersep=5pt,
	showspaces=false,
	showstringspaces=false,
	showtabs=false,
	tabsize=2,
	language=Python
}

\def\BibTeX{{\rm B\kern-.05em{\sc i\kern-.025em b}\kern-.08em
		T\kern-.1667em\lower.7ex\hbox{E}\kern-.125emX}}

\newenvironment{proof}{\emph{Proof.}}{\hfill$\square$}
\newtheorem{theorem}{Theorem}
\newtheorem{lemma}[theorem]{Lemma}

\newtheorem{corollary}{Corollary}
\newtheorem{definition}{Definition}
\newtheorem{proposition}{Proposition}

\usepackage[most]{tcolorbox}

\tcbset{
	algostyle/.style={
		colback=white,
		colframe=black,
		boxrule=0.8pt,
		arc=0mm,
		left=3pt,
		right=3pt,
		top=3pt,
		bottom=3pt,
		fonttitle=\bfseries
	}
}

\begin{document}

\title{Plausible Deniability in Fully Homomorphic Computation}

\author{%
	Shahzad Ahmad~\orcidlink{0000-0002-9654-869X}\thanks{Shahzad Ahmad and Stefan Rass are with the LIT Secure and Correct Systems Lab, Johannes Kepler University, Linz, Austria and Zahra Seyedi is associated with Computer Science and Engineering Department, Koç University, Istanbul, Türkiye (email: shahzad.ahmad@jku.at, stefan.rass@jku.at, zseyedi@ku.edu.tr).}
	\and Stefan Rass~\orcidlink{0000-0003-2821-2489}
	\and Zahra Seyedi~\orcidlink{0009-0002-8492-4640}}

\maketitle

\begin{abstract}
	We introduce \emph{Plausible Deniability in Fully Homomorphic Computation} (PD-FHC), a framework enabling users to outsource Boolean computations to an untrusted  cloud while maintaining both computational privacy against  honest-but-curious providers and plausible deniability  against coercive adversaries. We define the notion of a  \emph{Deniable Computation Medium} (DCM) and a  \emph{Deniable Computation Scheme} (DCS) as medium-independent abstractions, then instantiate them 
	using RGB images with Fredkin-gate circuits. One real circuit and several decoys share a single fixed Fredkin-gate wiring. Embedded control bits decide what each gate computes at each pixel, so the same wiring evaluates the real function at the real positions and decoy functions elsewhere. The cloud applies this one wiring to every pixel identically, processing all circuits in a single pass. Under coercion, the user reveals a decoy with verifiable results while the real circuit stays hidden. We formalize multi-round coercion games with existence and circuit-discovery advantages. For the image instantiation we prove \emph{information-theoretic position privacy} under a \emph{matched-marginal condition}: when the real, decoy, and fill bits are drawn from a common per-position law and placed at random, the embedded LSB plane is exchangeable, so an honest-but-curious provider gains no advantage over guessing at locating the real positions, for any such law and not only the uniform one. We are explicit that this is a condition Alice enforces, that it is distinct from steganalytic undetectability, and that the latter requires the embedded law to match the declared service's legitimate-input law. Our Python implementation, benchmarked across circuit sizes (5--302 gates) and image dimensions ($128^2$ to $512^2$), shows competitive performance with TFHE for Boolean circuits while providing deniability that FHE cannot natively offer.
\end{abstract}


\section{Introduction}\label{sec:intro}
\IEEEPARstart{O}{utsourcing} computation to cloud providers exposes sensitive data to untrusted infrastructure. Fully 
Homomorphic Encryption (FHE)~\cite{gentry_fully_2009} and Trusted Execution Environments 
(TEEs)~\cite{costan_intel_2016} protect data confidentiality but fail to address a distinct threat: \emph{coercion}. A powerful adversary who compels a user to explain her cloud activity will find that FHE ciphertexts confirm the use of encryption, and TEE attestation logs confirm that 
a specific computation was executed. Neither technology allows the user to credibly deny her true computational intent. This paper introduces \emph{Plausible Deniability in Fully Homomorphic Computation} (PD-FHC), a framework enabling users to outsource Boolean computations while 
maintaining the ability to deny, under coercion, what they actually computed.

\subsection{Motivating Scenario}\label{sec:scenario}
A photo lab routinely submits large batches of images to a cloud service that applies a per-pixel filter as part of its dithering and quality-control pipeline. This activity is entirely ordinary, and the per-pixel changes are visually imperceptible at the level of single-bit adjustments in each channel. A technician needs to privately check whether a confidential 8-bit measurement $v$ exceeds a threshold $\tau$. She must also be able to deny that this specific computation ever occurred if she is later compelled to explain her cloud use. Neither FHE (conspicuous ciphertexts) nor a TEE (an undeniable attestation log) can provide this. With PD-FHC, she turns the routine job into a deniable computation.

The technician embeds the input bits for her comparison into the least significant bits (LSBs) of chosen pixels, and embeds the input bits for several decoy comparisons (a brightness check, a color-balance check, a noise-level check) at other pixels. The declared cloud filter is a single fixed wiring of universal Fredkin gates, applied to every pixel identically. The control bits embedded alongside each pixel's data decide what each gate does at that position, so the same wiring evaluates her real comparison at the real positions and a different (decoy) function at every decoy position. The cloud reads three LSBs, applies the gate, and writes three LSBs, and it does this for every pixel. We stress what the cloud does \emph{not} do here: it does not brighten, balance, or denoise the image in the ordinary sense. The only operation is the pixel-wise Fredkin pass, and ``brightness check'' or ``color-balance check'' name the \emph{functions some of the embedded circuits compute}, not image edits the cloud performs. If the technician is later coerced, she reveals one of the decoy circuits and its verifiable result. The coercer cannot tell it apart from the real one.

This example illustrates two requirements that existing secure computation techniques fail to satisfy 
simultaneously:

\begin{enumerate}
	\item \textbf{Existence hiding (against an external observer).} To anyone watching the wire but not running the job, the interaction must look like ordinary cloud storage or image filtering, with no sign that a private computation is taking place. FHE fails here: sending ciphertexts to a cloud is conspicuous, and a coercive adversary can demand decryption keys. We are precise about scope: existence is hidden from an external eavesdropper, but \emph{not} from the cloud provider, who necessarily knows that \emph{some} computation runs because the provider is the one running it (see Section~\ref{sec:entities}).
	
	\item \textbf{Partial circuit hiding (against the cloud and external observers).} Even an observer who knows a computation occurred should not learn \emph{which} function was computed. The wiring is visible to the cloud, since the cloud needs it to compute, but the meaning of each gate is fixed by the embedded control bits and is not revealed. The user can therefore present a different, verifiable function as the one she ran. TEEs fail here: attestation logs confirm exactly what was executed. We call this \emph{partial} hiding because the gate types are deniable while the wiring stays visible. Hiding the wiring as well (\emph{full circuit hiding}) is a natural next step that we leave to future work.
\end{enumerate}

No existing technique meets both requirements at once. FHE~\cite{gentry_fully_2009,brakerski_efficient_2011} hides data values, but its very use is conspicuous, and a coercer can compel key disclosure. TEEs~\cite{costan_intel_2016,kaplan_amd_2016} produce attestation logs that are undeniable proof of what ran, 
and they have repeatedly leaked through side channels~\cite{bulck_foreshadow_nodate,moghimi_copycat_nodate}. MPC~\cite{Yao1982ProtocolsFS,Goldreich1987HowTP} requires interaction among multiple parties and is itself a conspicuous activity. We note that MPC does carry a natural form of \emph{robustness} against coercion: a coercer who forces one party to reveal its share still learns nothing unless it convinces a threshold number of parties to cooperate. That robustness is not deniability, though. The participants cannot deny that the computation took place or substitute a different one. Deniable encryption~\cite{canetti_deniable_1997,oneill_bi-deniable_2011} addresses what a message decrypts to, not what a user computed, and deniable storage such as hidden volumes~\cite{truecrypt,mcdonald_stegfs_2000} hides stored data rather than dynamic computational intent. PD-FHC targets the gap these leave open: hiding both the existence and the intent of an outsourced computation.

\subsection{Design Requirements}\label{sec:requirements}
We identify four properties that any deniable outsourced computation scheme must satisfy: (1)~\emph{uniform processability}  the cloud applies identical operations to every data element; (2)~\emph{sufficient capacity} the cover medium has far more addressable positions than  secret bits ($n \gg L \cdot \ell$); (3)~\emph{distribution preservation} gate operations preserve the statistical profile of non-secret positions; and (4)~\emph{cover plausibility} transmitting the medium to a cloud provider has a legitimate, non-suspicious purpose. We formalise these in Section~\ref{sec:framework} as properties of a \emph{Deniable Computation Medium} (DCM) and present PD-FHC as both a general framework satisfying these requirements and a concrete image-based instantiation.

\subsection{Contributions}

\begin{enumerate}
	\item \textbf{Abstract framework.} We define the notion of a Deniable Computation Medium and a Deniable Computation Scheme, separating the security model from any specific cover medium (Section~\ref{sec:framework}).
	
	\item \textbf{Two-phase threat model with multi-round coercion.} We formalize a threat model separating computational privacy (against an honest-but-curious cloud) from plausible deniability (against a coercive adversary who may demand multiple revelations), and define novel coercion games with existence and circuit-discovery advantages (Section~\ref{sec:threat}).
	
	\item \textbf{Image-based instantiation.} We instantiate the framework using RGB images and 
	Fredkin-gate circuits, proving information-theoretic position privacy under a matched-marginal assumption and bounding the coercive adversary's advantage (Sections~\ref{sec:protocol}--\ref{sec:security}).
	
	\item \textbf{Implementation and evaluation.} We provide a Python implementation benchmarked across circuit sizes (5--302 gates) and image dimensions ($128^2$ to $512^2$), with quantitative comparison against TFHE for equivalent Boolean circuits (Section~\ref{sec:eval}).
\end{enumerate}

\subsection{Paper Organization}
Section~\ref{sec:related} surveys related work. Section~\ref{sec:prelim} covers Fredkin gates and LSB 
steganography. Section~\ref{sec:framework} defines the abstract framework. Section~\ref{sec:threat} formalizes the threat model. Section~\ref{sec:protocol} presents the image-based protocol. Section~\ref{sec:security} proves security. Section~\ref{sec:eval} evaluates the implementation. Section~\ref{sec:discussion} discusses limitations. Section~\ref{sec:conclusion} concludes.

\section{Related Work}\label{sec:related}

\subsection{Fully Homomorphic Encryption}
Gentry's construction~\cite{gentry_fully_2009} enabled arbitrary computation on encrypted data. Subsequent schemes~\cite{brakerski_efficient_2011,openfhe2022,seal} have improved efficiency, and modern libraries support both Boolean and arithmetic circuits. However, FHE provides no mechanism for 
plausible deniability. The use of FHE is itself conspicuous: transmitting ciphertexts to a cloud provider is not normal behavior, and a coercive adversary can demand decryption keys. PD-FHC addresses this by embedding computation within a legitimate interaction whose very existence is non-suspicious.

\subsection{Trusted Execution Environments}
Hardware-based solutions such as Intel SGX~\cite{costan_intel_2016}, and AMD SEV~\cite{kaplan_amd_2016} provide isolated execution. Despite strong integrity guarantees, TEEs require specialized hardware, have proven vulnerable to side-channel attacks~\cite{aciicmez_power_2006,brasser_software_nodate, bulck_foreshadow_nodate,moghimi_copycat_nodate}, and produce attestation logs that constitute undeniable proof of execution. PD-FHC is purely algorithmic and hardware-agnostic, producing no attestation trail.

\subsection{Steganography and Homomorphic Steganography}
Classical steganography~\cite{fridrich2009} conceals the existence of a message within a cover medium. Advanced techniques by Holub and Fridrich~\cite{holub_designing_2012} and Pevn\'{y} et al.~\cite{pevny_using_2010} achieve high undetectability but do not support computation on embedded 
data. ProSt~\cite{ahmad2026prost} introduced homomorphic steganography, enabling Fredkin-gate computation on LSB-embedded data with information-theoretic privacy against an honest-but-curious cloud. However, ProSt embeds a single computation and does not address coercion. PD-FHC extends ProSt with multi-location embedding and a formal deniability model, enabling users to provide alternative explanations under coercion.

\subsection{Deniable Cryptography and Storage}
Canetti et al.~\cite{canetti_deniable_1997} introduced deniable encryption, allowing senders to produce fake randomness that makes ciphertexts appear to encrypt different plaintexts. Subsequent work addressed receiver deniability~\cite{klonowski2008receiver}, and bi-deniability~\cite{oneill_bi-deniable_2011}. Deniable storage systems such as hidden volumes (e.g., 
VeraCrypt/TrueCrypt) and StegFS~\cite{mcdonald_stegfs_2000} apply deniability to file storage: hidden volumes are indistinguishable from random data in free space. PD-FHC generalizes deniability from static data (messages, files) to dynamic \emph{computational intent}: Alice denies not what she stored, but what she \emph{computed}.

\subsection{Secure Multi-Party Computation}
MPC protocols~\cite{Yao1982ProtocolsFS,Goldreich1987HowTP} enable joint computation with input privacy. However, MPC requires interactive protocols among multiple parties, assumes honest participation, and is itself conspicuous. It does offer a form of robustness against coercion, since a coercer who forces one party to reveal its share learns nothing until enough parties are compelled. That is robustness, not deniability: the parties cannot deny that the computation happened or claim a different one. PD-FHC operates in a non-interactive, single-cloud setting and provides deniability that MPC cannot.

\subsection{Positioning}
Table~\ref{tab:comparison} compares PD-FHC with related approaches. PD-FHC is the only scheme that simultaneously provides computational privacy, existence hiding, and partial circuit hiding.

\begin{table}[t]
	\centering
	\caption{Comparison with related work.}
	\label{tab:comparison}
	\small
	\begin{tabular}{lccccc}
		\toprule
		Approach & \rotatebox{70}{Comp.\ Privacy} 
		& \rotatebox{70}{Exist.\ Hiding} 
		& \rotatebox{70}{Part.\ Circ.\ Hiding} 
		& \rotatebox{70}{Overhead} 
		& \rotatebox{70}{HW-Dep.} \\
		\midrule
		FHE       & Yes & No  & No     & V.\ High & No  \\
		TEE       & Yes & No  & No     & Low      & Yes \\
		ProSt     & Yes & Yes & No     & Medium   & No  \\
		Den.\ Enc.& Yes & --- & Yes$^*$& Low      & No  \\
		MPC       & Yes & No  & No$^\dagger$ & High & No  \\
		\textbf{PD-FHC} 
		& \textbf{Yes} & \textbf{Yes} 
		& \textbf{Yes} & \textbf{Med.} 
		& \textbf{No} \\
		\bottomrule
	\end{tabular}

	\smallskip
	{\footnotesize $^*$Message deniability only, not which function was computed.\\
		$^\dagger$MPC has natural robustness against forced revelation of a single party's share, but does not let the parties deny that the computation occurred or substitute a different one.}
\end{table}

\section{Preliminaries}\label{sec:prelim}

\subsection{LSB Steganography}
An RGB image is represented as a tensor $I \in \mathbb{Z}_{256}^{h \times w \times 3}$, where 
$h, w$ are height and width, and $3$ denotes the RGB channels. We embed and extract single bits via the least significant bit:
\begin{align}
	\Emb(I, b, (r,c,k)) &: 
	I[r,c,k] \leftarrow (I[r,c,k] \;\&\; \neg 1) \,|\, b 
	\label{eq:emb} \\
	\Ext(I, (r,c,k)) &: 
	\text{return } I[r,c,k] \;\&\; 1 
	\label{eq:ext}
\end{align}
For an $h \times w \times 3$ image, there are $n = h \cdot w \cdot 3$ single-bit positions. 
Each position $(r,c,k)$ identifies a unique row, column, and color channel, and carries exactly one LSB.

\subsection{Notation}
We write $n$ for the total number of addressable bit-positions in a medium instance (for an image, the $n$ LSB positions just defined), $L$ for the number of embedded circuits (one real, $L - 1$ decoys), $\ell$ for the number of input bits per circuit, $j^*$ for the index of the real circuit (known only to Alice), and $t$ for the number of coercion rounds. Each circuit's input bits occupy a position set $\Lambda_j \subset \mathcal{P}$, where $\mathcal{P}$ is the set of all bit-positions, and the position sets are pairwise disjoint, $\Lambda_i \cap \Lambda_j = \emptyset$ for $i \neq j$. We use $L$ only for the \emph{count} of circuits and $\Lambda_j$ for the \emph{sets} of positions, and never overload one symbol for both. We write $\pi$ for the per-position bit law that governs every non-secret position: the fill bits Alice writes into unused positions are drawn i.i.d.\ from $\pi$, and the matched-marginal assumption of Section~\ref{sec:security} asks the real and decoy input bits to follow the same $\pi$. A function $\nu : \mathbb{N} \to \mathbb{R}^+$ is \emph{negligible} if for every $c > 0$ there exists $n_0$ such that $\nu(n) \leq n^{-c}$ for all $n \geq n_0$. We use this notion only for the abstract framework. For the image instantiation we are deliberately more concrete: the relevant advantages are either exactly zero under a stated matched-distribution assumption, or controlled by an explicit statistical distance, and we say so in each case rather than appealing to asymptotic negligibility.

\subsection{Function Specification and Circuits}\label{sec:funcspec}
The private computation Alice wants is a Boolean function $f : \{0,1\}^\ell \to \{0,1\}^*$, represented by a circuit $C$ composed exclusively of universal Fredkin gates. Any standard Boolean formula in terms of conjunctions, disjunctions, and negations is algorithmically convertible into a Fredkin-gate-only representation~\cite{ahmad2026prost}, in the same way a circuit is, without loss of generality, representable using only NAND or only NOR gates. Our choice of the Fredkin gate has technical reasons that become clear in Section~\ref{sec:fredkin}.

To deny the computation of $f$ against the curious cloud provider, Alice considers a family of \emph{decoy} circuits $C_1, \ldots, C_L$ that differ from the circuit for $f$ \emph{only} in the setting of the control wires of the Fredkin gates. A Fredkin gate with a different control setting becomes a different elementary operation (AND, OR, NOT, and so on). The $j$-th decoy circuit therefore has the \emph{same components and the same wiring} as the circuit for $f$, but different gate semantics, determined by that circuit's inputs. Those inputs are the bits embedded in the least significant bits of the image pixels. Processing the LSBs of the images in one parallel pass thus evaluates the circuit for $f$ and all decoy circuits at once. This shared-wiring construction is what later lets Alice deny $f$ by pointing to a decoy.

We write $C_{\mathsf{obf}}$ for the specification of the circuit that contains only the wiring and is silent about the meaning of each Fredkin gate in it. For the purpose of computation this is \emph{sufficient} for Carol to evaluate Alice's desired function, since each Fredkin gate gets its concrete meaning from the bits encoded in the LSB plane. Because many such meanings coexist in the same wiring besides the real one, the real circuit is only one among many decoys, which already provides a natural form of obfuscation at no extra cost (we return to this in Section~\ref{sec:cobf}). We leave the general problem of obfuscating circuits out of scope, in light of the impossibility results of~\cite{barak_impossibility_2012}, and restrict our use of the word obfuscation to the hiding of gate functionality implied by the steganographically hidden inputs.

\subsection{Fredkin Gate}\label{sec:fredkin}
The Fredkin gate~\cite{fredkin_conservative_2002} is a three-input, three-output reversible gate defined as:
\[
F(c, x, y) = 
\begin{cases}
	(c, x, y) & \text{if } c = 0 \\
	(c, y, x) & \text{if } c = 1
\end{cases}
\]
where $c$ is the control bit and $x, y$ are data bits.

\noindent\textbf{Why Fredkin gates.} Three properties make the Fredkin gate a good fit for PD-FHC.

\smallskip\noindent (1)~\emph{Bit-level compatibility.} Each gate operates on individual bits, directly compatible with LSB steganography where each position carries one bit.

\smallskip\noindent (2)~\emph{Programmability through the control bit.} The same gate computes different elementary operations depending on its control bit. This is exactly what lets one fixed wiring serve as both the real circuit and every decoy: the control bits, supplied through the LSB plane, select the per-gate semantics. The Fredkin gate is not the only member of a suitable gate family. The Toffoli gate is another universal reversible gate that would fit Definition~\ref{def:dcm}. We use Fredkin gates throughout because their conditional-swap action also preserves the bit count at each position, which we exploit in the security analysis (Lemma~\ref{lem:hamming}, Section~\ref{sec:security}).

\smallskip\noindent (3)~\emph{Universality.} Any Boolean function can be computed using Fredkin gates with ancillary constant bits:
\begin{align}
	\mathsf{NOT}(x) &: F(1, x, 0) = (1, 0, x) 
	\implies \neg x \text{ on wire 3} \label{eq:not} \\
	\mathsf{AND}(x,y) &: F(x, y, 0) = (x, \_, x \wedge y) 
	\implies x \wedge y \text{ on wire 3} \label{eq:and}
\end{align}
$\mathsf{OR}$ follows from De Morgan's law using three gates. For example, the Boolean expression 
$(A \wedge C) \vee (\neg A \wedge B) \vee (\neg B \wedge \neg C)$ decomposes into 14 Fredkin gates 
covering negations, conjunctions, De Morgan disjunctions, and the ancillary wire routing needed to fan out shared inputs.

\section{Framework for Deniable Outsourced Computation}\label{sec:framework}
We define deniable outsourced computation as a general framework, independent of any specific cover medium. This abstraction separates the security model from the instantiation, allowing us to state security guarantees that hold for any medium satisfying our requirements. The framework is not tied to images: uncompressed audio and network packet payloads also fit it, as we note in Section~\ref{sec:instantiations}. Section~\ref{sec:protocol} then instantiates the framework in full detail with RGB images.

\subsection{Deniable Computation Medium}
A deniable computation scheme requires a \emph{cover medium} that serves two purposes: it provides a large space of addressable bit-positions for embedding secret computations among indistinguishable fill positions, and it supplies a legitimate context for transmitting data to a cloud provider.

\begin{definition}[Deniable Computation Medium]\label{def:dcm}
	A \emph{Deniable Computation Medium} (DCM) is a tuple $\mathcal{M} = (\mathcal{D}, \mathcal{P}, n, 
	\Read, \Write, \mathcal{G})$ where:
	\begin{itemize}
		\item $\mathcal{D}$ is the data domain, i.e., the set of valid medium instances;
		\item $\mathcal{P}$ is a position space with $|\mathcal{P}| = n$, indexing the addressable 
		bit-positions within any instance $M \in \mathcal{D}$;
		\item $\Read : \mathcal{D} \times \mathcal{P} \to \{0,1\}$ extracts the bit at a given position;
		\item $\Write : \mathcal{D} \times \mathcal{P} \times \{0,1\} \to \mathcal{D}$ sets the bit at a given position, leaving all other positions unchanged;
		\item $\mathcal{G}$ is a gate family. Each gate $G \in \mathcal{G}$ is a function $G : \{0,1\}^{\kappa} \to \{0,1\}^{\kappa}$ for some fixed arity $\kappa$. A gate is applied \emph{across all positions in parallel} by evaluating $G$ at every position $p \in \mathcal{P}$: given $\kappa$ input instances $M_1, \ldots, M_{\kappa} \in \mathcal{D}$, the gate produces $\kappa$ output instances $M'_1, \ldots, M'_{\kappa}$ such that for each 
		$p \in \mathcal{P}$:
		\[
		\bigl(\Read(M'_1, p), \ldots, 
		\Read(M'_{\kappa}, p)\bigr) = 
		G\bigl(\Read(M_1, p), \ldots,\] 
		\[\Read(M_{\kappa}, p)\bigr).
		\]
	\end{itemize}
\end{definition}

\noindent The uniform application of gates is the mechanism that ensures the cloud provider treats every position identically: the same function $G$ is evaluated at every $p \in \mathcal{P}$, regardless of whether $p$ carries secret data, decoy data, or noise.

\subsection{Security Properties of a DCM}\label{sec:dcm-properties}
Not every medium is suitable for deniable computation. We require four properties.

\begin{definition}[Suitable DCM]\label{def:suitable-dcm}
	A DCM $\mathcal{M}$ is \emph{suitable for deniable computation} if it satisfies the following properties:
\end{definition}

\smallskip\noindent\textbf{Property~1 (Uniform Processability).} For every gate $G \in \mathcal{G}$ and every set of input instances $M_1, \ldots, M_{\kappa}$, the uniform application of $G$ transforms each position $p \in \mathcal{P}$ independently and identically. The computational cost and memory-access pattern of evaluating $G$ are identical for every position, preventing the cloud 
provider from inferring which positions carry secret data through timing or cache side-channels.

\smallskip\noindent\textbf{Property~2 (Sufficient Capacity).} The number of addressable positions $n$ satisfies $n \gg L \cdot \ell$, where $L$ is the number of embedded embedded circuits and $\ell$ is the number of bit-positions required per circuit. This ensures that secret positions constitute a negligible fraction of the total, making them difficult to locate.

\smallskip\noindent\textbf{Property~3 (Distribution Preservation).} Let $\pi$ denote a distribution over $\{0,1\}$. The choice of $\pi$ is not fixed by the framework. What matters is that one common $\pi$ governs every position that does not carry the real circuit's data, that is, the decoy positions and the fill positions alike. If all such non-real positions in the input instances carry i.i.d.\ bits drawn from $\pi$, then after applying any sequence of gates $G_1, \ldots, G_m \in \mathcal{G}$, the bits at those positions in the output instances remain distributed according to $\pi$:
\[
\Read(M''_i, p) \sim \pi \quad \text{for all non-secret } p \] 
\[\text{ and all output instances } M''_i.
\]
This prevents the provider from identifying positions where ``structured'' computation occurred by comparing input and output distributions.

\smallskip\noindent\textbf{Property~4 (Cover Plausibility).} There exists a legitimate, non-suspicious use case for transmitting instances of $\mathcal{D}$ to a cloud provider for processing and receiving processed instances in return. The gate specifications in $\mathcal{G}$ must be consistent with this use case, so that the cloud provider's view of the protocol is indistinguishable from a routine service interaction.

\smallskip
\noindent\textbf{Remark.} Properties~1--3 are formal and verifiable for a given instantiation. Property~4 is inherently contextual: it depends on the operational setting and cannot be captured by a cryptographic definition. We discuss how specific instantiations satisfy Property~4 qualitatively in Section~\ref{sec:protocol}.

\subsection{Deniable Computation Scheme}
We now define the algorithmic structure of a deniable computation over a DCM.

\begin{definition}[Deniable Computation Scheme]\label{def:scheme}
	A \emph{Deniable Computation Scheme} (DCS) over a suitable DCM $\mathcal{M}$ is a tuple of algorithms $\Pi = (\LocGen, \Emb, \Comp, \Ext)$:
	\begin{itemize}
		\item $\{\Lambda_1, \ldots, \Lambda_L\} \leftarrow \LocGen(\mathcal{P}, L, \ell)$: Samples $L$ pairwise disjoint subsets of $\mathcal{P}$, each of size $\ell$, uniformly at random ($\Lambda_i \cap \Lambda_j = \emptyset$ for $i \neq j$). These position sets are known only to Alice.
		
		\item $\mathbf{M}' \leftarrow \Emb(\mathbf{M}, \{(\Lambda_j, x_j)\}_{j=1}^{L})$: For each circuit $j$, writes input bits $x_j \in \{0,1\}^{\ell}$ at positions $\Lambda_j$. All positions $p \notin \bigcup_{j} \Lambda_j$ are filled with i.i.d.\ bits from $\pi$.
		
		\item $\mathbf{M}'' \leftarrow \Comp(C, \mathbf{M}')$: The cloud provider evaluates circuit $C = (G_1, \ldots, G_m)$ by applying each gate $G_i$ uniformly to all positions in $\mathcal{P}$. The provider receives only $C$ and $\mathbf{M}'$, with no knowledge of position sets or circuit assignments.
		
		\item $y_j \leftarrow \Ext(\mathbf{M}'', \Lambda_j)$: Alice reads output bits for circuit $j$ from positions $\Lambda_j$ in the computed instances.
	\end{itemize}
\end{definition}

\noindent\textbf{Correctness.} A DCS is correct if for every circuit $j \in \{1, \ldots, L\}$, the extracted output $y_j = \Ext(\mathbf{M}'', \Lambda_j)$ equals $f_j(x_j)$, where $f_j$ is the Boolean function computed by circuit $C$ when evaluated on the bits at positions $\Lambda_j$. The bits read from $\Lambda_j$ set the control wires of the gates, so they decide what each gate computes. This is what makes the per-circuit functions $f_j$ distinct from one another under one shared wiring, and it is what makes the decoys usable as alternative explanations.

\smallskip
\noindent\textbf{Uniform processing as structural invariant.} The critical constraint is that $\Comp$ applies every gate to \emph{every} position in $\mathcal{P}$. Combined with Property~3, this ensures 
that the output medium $\mathbf{M}''$ retains the same indistinguishability guarantees as $\mathbf{M}'$: the provider cannot distinguish real from fill positions by observing the computation.

\subsection{Possible Instantiations}\label{sec:instantiations}
The DCM framework is not limited to images. Any medium with high-capacity bit-addressable positions and a legitimate cloud-processing use case is a candidate: uncompressed audio (LSBs in 16-bit samples, cloud audio normalization), or network packet payloads (cloud-based traffic processing). We focus on RGB images (Section~\ref{sec:protocol}), where $\mathcal{P} = \{(r,c,k)\}$ indexes pixel-channel LSBs, the gate family is Fredkin gates ($\kappa = 3$), and cover plausibility follows from the prevalence of cloud image processing services.

\section{Threat Model and Security Definitions}\label{sec:threat}
We formalize the adversarial setting in which a Deniable Computation Scheme (Definition~\ref{def:scheme}) operates. Our threat model separates two adversaries with distinct 
capabilities and goals, reflecting realistic settings where the entity executing a computation differs from the entity demanding explanations.

\subsection{Entities}\label{sec:entities}
We distinguish two kinds of attacker: external parties and an honest-but-curious cloud provider. By \emph{existence hiding} we mean that to an external eavesdropper the cloud looks like a mere storage or filtering service, while it is in fact also performing a computation. Against the cloud provider, existence is \emph{not} hidden, because the provider is the one running the computation and necessarily knows that some computation occurs. What stays hidden from the cloud is the exact data being processed and which function among the family is being evaluated. Table~\ref{tab:adversaries} summarises this split, which the rest of the section makes precise.

\begin{table}[t]
	\centering
	\caption{What each adversary can and cannot tell. Existence hiding is meaningful only against an external observer. The cloud always knows a computation runs.}
	\label{tab:adversaries}
	\small
	\resizebox{\columnwidth}{!}{
	\begin{tabular}{llcc}
		\toprule
		Adversary & Type & Existence hidden? & Function hidden? \\
		\midrule
		External       & Eavesdropper        & Yes & Yes (deniable) \\
		Cloud provider & Honest-but-curious  & No  & Yes \\
		\bottomrule
	\end{tabular}
}
\end{table}

\smallskip\noindent\textbf{Alice (Client).} Alice possesses private inputs and wishes to outsource a Boolean computation to a cloud provider. She prepares $L$ embedded circuits (one real at secret index $j^*$, and $L - 1$ decoys) using a DCS $\Pi$ over a DCM $\mathcal{M}$. Alice knows all position sets $\{\Lambda_1, \ldots, \Lambda_L\}$, the index $j^*$, all circuit descriptions, and all inputs and outputs for 
every circuit.

\smallskip\noindent\textbf{Carol (Cloud Provider).} An honest-but-curious (HBC) adversary. Carol receives embedded medium instances $\mathbf{M}'$ and the wiring specification $C_{\mathsf{obf}}$ (Section~\ref{sec:funcspec}), and executes the protocol faithfully. She may attempt to infer secret information from her observations. Carol knows the DCS algorithms, the gate family $\mathcal{G}$, and all public parameters. She does \emph{not} know: any position sets $\Lambda_j$, the number of circuits $L$, the true index $j^*$, or which 
positions carry secret data versus fill. We grant Carol the additional power of a \emph{steganalytic warden}: she may run any detector (RS or sample-pair analysis~\cite{fridrich2009}, or a learned detector~\cite{boroumand2019srnet}) on $\mathbf{M}'$ and $\mathbf{M}''$ to test whether they deviate from the distribution she expects a legitimate input to her service to have. This capability is the subject of Section~\ref{sec:steganalysis}.

\smallskip\noindent\textbf{Eve (External Coercive Adversary).} A computationally unbounded external adversary who can compel Alice to reveal information. Eve models legal authorities, state actors, or any entity with the power to demand explanations under threat of consequences. Eve can observe the complete communication transcript between Alice and Carol, $V = (\mathbf{M}', C_{\mathsf{obf}}, \mathbf{M}'')$, where $\mathbf{M}'$ is the embedded input medium that Alice sends and $\mathbf{M}''$ is the processed medium that Carol returns. 
Under coercion, Eve can demand that Alice reveal embedded circuits, including their position sets, 
circuit descriptions, inputs, and outputs.

\emph{Critically, Eve cannot force Alice to reveal information that Alice denies possessing.} Eve does not learn the exact number $L$ of embedded circuits from the transcript, though she can upper-bound it from the carrier capacity (Section~\ref{sec:game-discussion}). After Alice reveals $t$ circuits and claims no further circuits exist, Eve cannot prove this false, because the unrevealed real-and-decoy positions are exchangeable with the surrounding fill positions and so statistically indistinguishable from them (proved in Theorem~\ref{thm:existence}). This is analogous to hidden-volume deniability in deniable 
storage systems~\cite{truecrypt}: an adversary who demands all volumes cannot prove that a hidden volume exists if the free space is filled with random data.

\smallskip\noindent\textbf{Active attackers are out of scope.} Both guarantees fail against an \emph{active} attacker. An active cloud provider could refuse to evaluate the circuit, or evaluate a different one. An active external attacker could apply ordinary image manipulations that overwrite the LSB plane and so destroy the embedded circuit before it is computed. We treat only honest-but-curious processing and passive observation, and we return to active deviation and collusion in Section~\ref{sec:discussion}.

\subsection{Adversary Separation}\label{sec:separation}
The security model relies on a separation between Carol and Eve:

\smallskip\noindent\emph{Temporal.} Carol processes data during computation. Eve coerces Alice after computation completes (or at a later time).

\smallskip\noindent\emph{Capability.} Carol has infrastructure access (she processes the medium) but no coercion power. Eve has coercion power but no direct access to Carol's infrastructure or intermediate computation states.

\smallskip\noindent\emph{Non-collusion.} Carol and Eve do not share information or coordinate. We discuss collusion attacks and mitigations in Section~\ref{sec:discussion}.

\subsection{Security Game: Computational Privacy}\label{sec:game-privacy}
We formalize what it means for the cloud provider to learn nothing about secret data. The game is defined over the abstract DCM (Definition~\ref{def:dcm}), not any specific medium.

\begin{definition}[Computational Privacy Game]\label{def:privacy-game}
	Let $\Pi$ be a DCS over a DCM $\mathcal{M}$ with $n = |\mathcal{P}|$ positions. The game 
	$\mathsf{Game}_{\mathsf{priv}}^{\mathcal{A}}(n, L, \ell)$ between a challenger and adversary $\mathcal{A}$ (modeling Carol) proceeds as follows:
	\begin{enumerate}
		\item \textbf{Setup.} The challenger generates $L$ disjoint position sets 
		$\{\Lambda_1, \ldots, \Lambda_L\} \leftarrow \LocGen(\mathcal{P}, L, \ell)$ and embeds $L$ 
		circuits with inputs $\{x_1, \ldots, x_L\}$, producing $\mathbf{M}' \leftarrow 
		\Emb(\mathbf{M}, \{(\Lambda_j, x_j)\}_{j=1}^{L})$. All non-secret positions carry i.i.d.\ bits from $\pi$.
		
		\item \textbf{Computation.} The challenger applies the wiring specification $C_{\mathsf{obf}}$ uniformly, producing $\mathbf{M}''$.
		
		\item \textbf{Challenge.} $\mathcal{A}$ receives $(\mathbf{M}', C_{\mathsf{obf}}, \mathbf{M}'')$ and outputs a position $p^* \in \mathcal{P}$.
		
		\item \textbf{Win condition.} $\mathcal{A}$ wins if $p^* \in \bigcup_{j=1}^{L} \Lambda_j$.
	\end{enumerate}
\end{definition}

\begin{definition}[Privacy Advantage]\label{def:priv-advantage}
	The adversary's advantage is:
	\[
	\Adv_{\mathsf{priv}}^{\mathcal{A}}(n, L, \ell) = \Pr[\mathcal{A}\ \text{wins}] - \frac{L \cdot \ell}{n}\]
	where $L \cdot \ell / n$ is the probability of hitting a secret position by random guessing. $\Pi$ provides \emph{computational privacy} if this advantage is negligible in~$n$ for all efficient adversaries $\mathcal{A}$. The quantity can be negative if $\mathcal{A}$ does worse than guessing. What we bound throughout is the advantage from above, and a value at or below $0$ means no useful attack.
\end{definition}

\noindent\textbf{Remark.} The game captures Carol's strongest attack: identify \emph{any} secret position. A weaker variant (guess the secret \emph{value} at a known position) is trivially subsumed.

\subsection{Security Game: Plausible Deniability}\label{sec:game-deniability}
We model coercion as an interactive, multi-round game where the adversary may demand repeated explanations, reflecting realistic coercion settings.

\begin{definition}[Multi-Round Coercion Game]\label{def:coercion-game}
	Let $\Pi$ be a DCS over a DCM $\mathcal{M}$ with $n = |\mathcal{P}|$. The game 
	$\mathsf{Game}_{\mathsf{deny}}^{\mathcal{A}}(n, L, t)$, parameterized by $L \geq 2$ circuits and $t < L$ coercion rounds, proceeds as follows:
	\begin{enumerate}
		\item \textbf{Setup.} Alice prepares $L$ circuits $\{s_1, \ldots, s_L\}$ where 
		$s_j = (C_j, x_j, \Lambda_j, y_j)$ consists of a circuit, inputs, position set, and correct output 
		$y_j = f_j(x_j)$. One circuit $s_{j^*}$ is Alice's real computation. All $L$ circuits share one fixed Fredkin-gate wiring and differ only in the control bits that set each gate's operation (Section~\ref{sec:funcspec}), so they have identical gate count, depth, and topology by construction, with no dummy-gate padding. Position sets are pairwise disjoint and uniformly random.
		
		\item \textbf{Embedding \& Computation.} Alice embeds all circuits and transmits the medium to Carol. Carol processes it uniformly, producing $\mathbf{M}''$. The adversary $\mathcal{A}$ (modeling Eve) observes the full transcript $V = (\mathbf{M}', C_{\mathsf{obf}}, \mathbf{M}'')$.
		
		\item \textbf{Coercion rounds.} For $i = 1, \ldots, t$:
		\begin{enumerate}
			\item $\mathcal{A}$ demands that Alice reveal a circuit.
			\item Alice reveals a decoy $s_{j_i} = (C_{j_i}, x_{j_i}, \Lambda_{j_i}, y_{j_i})$ 
			where $j_i \neq j^*$ and $j_i \notin \{j_1, \ldots, j_{i-1}\}$ (a fresh, previously 
			unrevealed decoy).
			\item $\mathcal{A}$ verifies: reads output bits from $\mathbf{M}''$ at positions $\Lambda_{j_i}$ and checks $y_{j_i} = f_{j_i}(x_{j_i})$.
		\end{enumerate}
		
		\item \textbf{Claim.} After $t$ rounds, Alice claims no further circuits exist.
		
		\item \textbf{Adversary output.} $\mathcal{A}$ outputs a bit $b' \in \{0,1\}$ (for existence distinguishing: $b' = 1$ if $\mathcal{A}$ believes unrevealed circuits remain) or an index $j'$ (for circuit discovery: $\mathcal{A}$'s guess of~$j^*$).
	\end{enumerate}
\end{definition}

\begin{definition}[Deniability Advantages]\label{def:deny-advantage}
	We define two advantages capturing distinct threats:
	
	\smallskip\noindent\emph{Existence distinguishing.} Consider two sub-experiments. In $\mathsf{Exp}_0$, Alice embeds exactly $t$ circuits (all of which are revealed). In $\mathsf{Exp}_1$, Alice embeds $t < L$ circuits and reveals $t$ of them (hiding $L - t$, including the real one). In both experiments, all non-revealed, non-secret positions are i.i.d.\ from $\pi$. Here $\mathcal{A}(\mathsf{Exp}_i)$ denotes the bit the adversary outputs when it faces experiment $i$, that is, its guess of which of the two situations actually occurred. The existence advantage is:
	\begin{equation}
		\Adv_{\mathsf{exist}}^{\mathcal{A}}(n, L, t) = \bigl|\Pr[\mathcal{A}(\mathsf{Exp}_1) = 1] 
		- \Pr[\mathcal{A}(\mathsf{Exp}_0) = 1]\bigr|
		\label{eq:adv-exist}
	\end{equation}
	
	\smallskip\noindent\emph{Circuit discovery.} Assuming $\mathcal{A}$ knows $L$ (a strictly stronger adversary), the circuit-discovery advantage measures the ability to identify $j^*$ among the $L - t$ unrevealed circuits:
	\begin{equation}
		\Adv_{\mathsf{intent}}^{\mathcal{A}}(n, L, t) = \Pr[j' = j^*] - \frac{1}{L - t}
		\label{eq:adv-intent}
	\end{equation}
	
	\noindent For the abstract framework we say $\Pi$ provides \emph{plausible deniability} if both advantages vanish as the relevant single-position advantage $\varepsilon(n)$ does. We are deliberately concrete for the image instantiation: there the existence advantage is exactly $0$ under the matched-marginal condition (Theorem~\ref{thm:existence}) and otherwise bounded by the statistical distance $\Delta$ of Section~\ref{sec:steganalysis}, while the circuit-discovery advantage is bounded as in Theorem~\ref{thm:intent}. We do not claim asymptotic negligibility, since none of these quantities is negligible in the cryptographic sense. They are either zero under a stated condition or controlled by an explicit, measurable distance.
\end{definition}

\subsection{Discussion of the Coercion Model}\label{sec:game-discussion}
\noindent\textbf{Two dimensions.} The game captures two distinct threats. \emph{Existence distinguishing} (Eq.~\ref{eq:adv-exist}) asks whether Eve can tell a world where Alice hides circuits from one where she has revealed them all, the direct analogue of hidden-volume detection in deniable storage. \emph{Circuit discovery} (Eq.~\ref{eq:adv-intent}) asks whether an Eve who somehow knows $L$ can pick out $j^*$ among the unrevealed circuits. Existence hiding is the stronger property: an Eve who cannot detect hidden circuits cannot identify $j^*$ among them either. We use the term \emph{partial circuit hiding} for the second guarantee, because the wiring is visible while only the gate semantics are concealed.

\smallskip\noindent\textbf{Why $t < L$ is necessary.} If $t \geq L$, Alice runs out of decoys and must either reveal the real computation or be caught in an inconsistency. This is inherent to any scheme that fixes its decoys in advance: Alice commits to a finite set of $L$ circuits at embedding time and cannot manufacture new consistent decoys once coercion starts. Alice's security parameter is $L - t$, the number of unrevealed circuits after coercion. Increasing $L$ provides more ``budget'' for deniability at the cost of additional embedding (which is cheap, since $L \cdot \ell \ll n$ for practical parameters).

\smallskip\noindent\textbf{What Eve can and cannot infer about $L$.} Eve does not learn the exact number $L$ of embedded circuits from the transcript. We are careful not to overstate this. From the carrier capacity $n$ and the per-circuit size $\ell$, Eve can compute an upper bound $L \le n/\ell$ on how many circuits could fit, and a large image leaves room for many. So Eve may reasonably suspect that more circuits exist and press Alice for further revelations. What she cannot do is \emph{prove} that any specific unrevealed circuit exists: after Alice reveals $t$ circuits and stops, the remaining $n - t\cdot\ell$ positions split into $(L-t)\cdot\ell$ unrevealed real-and-decoy positions and $n - L\cdot\ell$ fill positions, and under the matched-marginal condition (Property~3) these two classes are exchangeable and so statistically indistinguishable. Theorem~\ref{thm:existence} makes this precise. The residual pressure Eve can apply is a danger to Alice in practice, not a break of the formal guarantee.

\smallskip\noindent\textbf{Verification does not break deniability.} In each coercion round, Eve verifies that the revealed circuit is consistent: the output at the revealed positions matches the claimed circuit and inputs. This verification \emph{always succeeds} for both real and decoy circuits (by correctness of the DCS), so passing verification provides Eve with no distinguishing 
information.

\smallskip\noindent\textbf{Relationship to prior models.} Our existence advantage 
$\Adv_{\mathsf{exist}}$ is analogous to the hidden-volume detection advantage in deniable storage~\cite{canetti_deniable_1997,mcdonald_stegfs_2000}. Our circuit-discovery advantage 
$\Adv_{\mathsf{intent}}$ is analogous to the indistinguishability advantage in deniable encryption~\cite{canetti_deniable_1997}. PD-FHC combines both in a single framework for the computational (rather than storage or communication) setting.

\section{PD-FHC: Image-Based Instantiation}\label{sec:protocol}
We instantiate the abstract framework of Section~\ref{sec:framework} using RGB images as the cover 
medium and Fredkin gates as the gate family.

\subsection{Image-Based DCM}\label{sec:image-dcm}
We define the image-based DCM $\mathcal{M}_{\mathsf{img}}$ as an instance of Definition~\ref{def:dcm}:

\smallskip\noindent
\begin{tabularx}{\columnwidth}{@{}lX@{}}
	$\mathcal{D}$ &
	$= \mathbb{Z}_{256}^{h \times w \times 3}$ (RGB images with pixel values in $\{0,\ldots,255\}$) \\
	
	$\mathcal{P}$ &
	$= \{(r,c,k): 0 \le r < h,\; 0 \le c < w,\; 0 \le k < 3\}$ \\
	
	$n$ &
	$= h \cdot w \cdot 3$ (total LSB positions) \\
	
	$\Read(I,(r,c,k))$ &
	$= I[r,c,k] \,\&\, 1$ (extract LSB) \\
	
	$\Write(I,(r,c,k),b)$ &
	$= I[r,c,k] \leftarrow (I[r,c,k] \,\&\, \neg 1)\,|\,b$
	(set LSB) \\
	
	$\mathcal{G}$ &
	$= \{F\}$ where $F:\{0,1\}^3 \to \{0,1\}^3$
	is the Fredkin gate ($\kappa=3$)
\end{tabularx}
\smallskip\noindent We verify that $\mathcal{M}_{\mathsf{img}}$ satisfies the four properties of a suitable DCM (Definition~\ref{def:suitable-dcm}):

\smallskip\noindent\emph{Property~1 (Uniform Processability).} Each Fredkin gate reads one bit from 
each of three input images at position $p$ and writes one bit to each of three output images at the same position $p$. This is applied identically for every $p \in \mathcal{P}$, with identical computation and memory access per position.

\smallskip\noindent\emph{Property~2 (Sufficient Capacity).} Even a modest $128 \times 128$ image gives 
$n = 49{,}152$, supporting hundreds of circuits with $\ell = 16$ bits each while keeping $L \cdot \ell / n < 1\%$.

\smallskip\noindent\emph{Property~3 (Distribution Preservation).} By Lemma~\ref{lem:hamming} (Section~\ref{sec:security}) the Fredkin gate only permutes the three bits at a position, so it preserves whatever per-position law those bits follow. If the decoy and fill positions carry i.i.d.\ bits from a common law $\pi$, they still follow $\pi$ after any number of Fredkin evaluations. We stress what this does \emph{not} require. It does \emph{not} require $\pi = \mathrm{Ber}(1/2)$, and it does \emph{not} make any distributional assumption about Alice's real input data, which we must not constrain. The property we actually use is that the decoy and fill positions are drawn from the \emph{same} law as the real positions. Section~\ref{sec:security} makes this matched-marginal condition precise and shows it is what Alice controls.

\subsubsection{From Abstraction to Pixels: a Worked Example}\label{sec:worked}
To connect the abstract DCM (Definition~\ref{def:dcm}) to concrete pixels, consider a $2 \times 2$ RGB tile and the threshold function (a simple Boolean comparison) ``$x_1 \wedge x_2$'' with $\ell = 2$ input bits. The 
position space is $\mathcal{P} = \{(r,c,k)\}$ with $n = 2 \cdot 2 \cdot 3 = 
12$ LSB positions. Suppose $\LocGen$ selects $\Lambda_1 = \{(0,0,0),(1,1,2)\}$ for the real circuit and 
$\Lambda_2 = \{(0,1,1),(1,0,0)\}$ for a decoy. Embedding writes $x_1, x_2$ into the LSBs at $\Lambda_1$ via 
$\Write$, e.g.\ a red value $200 = 11001000_2$ carrying bit $1$ becomes $11001001_2 = 201$, a change 
invisible to the eye. The decoy input bits are written at $\Lambda_2$ in the same way, and every remaining LSB (outside $\Lambda_1$ and $\Lambda_2$) is a fresh CSPRNG bit. Carol then applies the Fredkin circuit to \emph{all 12 positions} identically: she has no way to tell that $(0,0,0)$ holds a real input bit while $(0,1,0)$ holds a fill bit, because both are single LSBs undergoing the same gate. Alice recovers the result by reading the LSBs at $\Lambda_1$ from the returned tile. Figure~\ref{fig:worked} draws this tile concretely, and Figure~\ref{fig:pipeline} shows the full pipeline. The key point to take away is that the abstract $\Read/\Write/\mathcal{G}$ of Definition~\ref{def:dcm} are nothing more than LSB extraction, LSB insertion, and pixel-wise Fredkin evaluation.

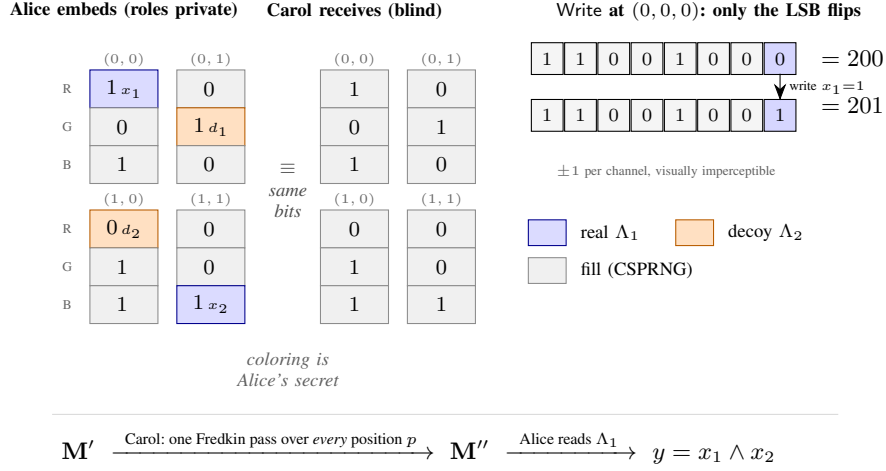
\begin{figure*}[t]
	\centering
	\begin{tikzpicture}[
		font=\small,
		>={Stealth[length=2.2mm]},
		cell/.style={draw, minimum width=0.9cm, minimum height=0.5cm,
			inner sep=0pt, font=\footnotesize},
		bitcell/.style={draw, minimum width=0.42cm, minimum height=0.42cm,
			inner sep=0pt, font=\scriptsize},
		tlabel/.style={font=\scriptsize\bfseries},
		plabel/.style={font=\tiny, text=black!60},
		]
		\colorlet{realf}{blue!14}\colorlet{realb}{blue!55!black}
		\colorlet{decoyf}{orange!24}\colorlet{decoyb}{orange!72!black}
		\colorlet{fillf}{black!6}\colorlet{fillb}{black!45}
		
		\begin{scope}[shift={(0,0)}]
			\node[tlabel] at (0.45,0.78) {Alice embeds (roles private)};
			\foreach \px/\py/\ch/\fc/\bc/\bit/\tag in {
				0/0/0/realf/realb/1/x_1, 0/0/1/fillf/fillb/0/{}, 0/0/2/fillf/fillb/1/{},
				1/0/0/fillf/fillb/0/{}, 1/0/1/decoyf/decoyb/1/d_1, 1/0/2/fillf/fillb/0/{},
				0/1/0/decoyf/decoyb/0/d_2, 0/1/1/fillf/fillb/1/{}, 0/1/2/fillf/fillb/1/{},
				1/1/0/fillf/fillb/0/{}, 1/1/1/fillf/fillb/0/{}, 1/1/2/realf/realb/1/x_2}
			{
				\pgfmathsetmacro{\xc}{\px*1.15+0.45}
				\pgfmathsetmacro{\yc}{-(\py*1.85+\ch*0.5+0.25)}
				\node[cell, fill=\fc, draw=\bc] at (\xc,\yc) {$\bit$\,{\tiny$\tag$}};
			}
			\foreach \px/\py in {0/0,1/0,0/1,1/1}{
				\pgfmathsetmacro{\xc}{\px*1.15+0.45}
				\pgfmathsetmacro{\yc}{-(\py*1.85)+0.14}
				\node[plabel] at (\xc,\yc) {$(\py,\px)$};
			}
			\node[plabel] at (-0.30,-0.25) {R};
			\node[plabel] at (-0.30,-0.75) {G};
			\node[plabel] at (-0.30,-1.25) {B};
			\node[plabel] at (-0.30,-2.10) {R};
			\node[plabel] at (-0.30,-2.60) {G};
			\node[plabel] at (-0.30,-3.10) {B};
		\end{scope}
		
		\begin{scope}[shift={(3.05,0)}]
			\node[tlabel] at (0.45,0.78) {Carol receives (blind)};
			\foreach \px/\py/\ch/\bit in {
				0/0/0/1,0/0/1/0,0/0/2/1, 1/0/0/0,1/0/1/1,1/0/2/0,
				0/1/0/0,0/1/1/1,0/1/2/1, 1/1/0/0,1/1/1/0,1/1/2/1}
			{
				\pgfmathsetmacro{\xc}{\px*1.15+0.45}
				\pgfmathsetmacro{\yc}{-(\py*1.85+\ch*0.5+0.25)}
				\node[cell, fill=black!6, draw=black!45] at (\xc,\yc) {$\bit$};
			}
			\foreach \px/\py in {0/0,1/0,0/1,1/1}{
				\pgfmathsetmacro{\xc}{\px*1.15+0.45}
				\pgfmathsetmacro{\yc}{-(\py*1.85)+0.14}
				\node[plabel] at (\xc,\yc) {$(\py,\px)$};}
		\end{scope}
		
		\node[font=\scriptsize, text=black!70] at (2.62,-1.30) {$\equiv$};
		\node[align=center, font=\scriptsize\itshape, text=black!70] at (2.62,-1.75) {same\\bits};
		\node[align=center, font=\scriptsize\itshape, text=black!62] at (2.62,-3.95) {coloring is\\Alice's secret};
		
		\begin{scope}[shift={(6.05,0.15)}]
			\node[tlabel, anchor=west] at (0,0.63) {$\Write$ at $(0,0,0)$: only the LSB flips};
			\foreach \i/\b/\c in {0/1/black!4,1/1/black!4,2/0/black!4,3/0/black!4,4/1/black!4,5/0/black!4,6/0/black!4,7/0/blue!16}{
				\pgfmathsetmacro{\xx}{\i*0.44}
				\node[bitcell, fill=\c] at (\xx,0) {$\b$};}
			\node[anchor=west] at (3.5,0) {$=200$};
			\foreach \i/\b/\c in {0/1/black!4,1/1/black!4,2/0/black!4,3/0/black!4,4/1/black!4,5/0/black!4,6/0/black!4,7/1/blue!16}{
				\pgfmathsetmacro{\xx}{\i*0.44}
				\node[bitcell, fill=\c] at (\xx,-0.75) {$\b$};}
			\node[anchor=west] at (3.5,-0.62) {$=201$};
			\draw[->] (3.08,-0.20) -- node[right, font=\tiny]{write $x_1{=}1$} (3.08,-0.55);
			\node[font=\tiny, text=black!62, anchor=west] at (0,-1.5) {$\pm1$ per channel, visually imperceptible};
		\end{scope}
		
		\begin{scope}[shift={(6.05,-2.15)}]
			\node[cell, fill=realf, draw=realb, minimum width=0.5cm, minimum height=0.34cm] at (0,0) {};
			\node[anchor=west, font=\scriptsize] at (0.32,0) {real $\Lambda_1$};
			\node[cell, fill=decoyf, draw=decoyb, minimum width=0.5cm, minimum height=0.34cm] at (1.95,0) {};
			\node[anchor=west, font=\scriptsize] at (2.27,0) {decoy $\Lambda_2$};
			\node[cell, fill=fillf, draw=fillb, minimum width=0.5cm, minimum height=0.34cm] at (0,-0.5) {};
			\node[anchor=west, font=\scriptsize] at (0.32,-0.5) {fill (CSPRNG)};
		\end{scope}
		
		\draw[black!15] (-0.5,-4.55) -- (9.9,-4.55);
		\node[anchor=west, font=\small] at (-0.5,-4.98)
		{$\mathbf{M}'\ \xrightarrow{\ \text{Carol: one Fredkin pass over \emph{every} position }p\ }\ \mathbf{M}''\ \xrightarrow{\ \text{Alice reads }\Lambda_1\ }\ y=x_1\wedge x_2$};
	\end{tikzpicture}
	\caption{Concrete instantiation of the worked example on a single $2\times2$ RGB tile ($n=12$ LSB positions; each pixel holds channels R,G,B top-to-bottom, indexed $(r,c,k)$). \emph{Left:} Alice embeds the real inputs $x_1,x_2$ at $\Lambda_1=\{(0,0,0),(1,1,2)\}$ and decoy inputs $d_1,d_2$ at $\Lambda_2=\{(0,1,1),(1,0,0)\}$; every remaining LSB is CSPRNG fill. \emph{Top right:} embedding writes one bit into one LSB, e.g.\ $200\to201$, a $\pm1$ change per channel that is visually imperceptible. \emph{Right:} Carol receives the identical bit-string with no role information, so real, decoy, and fill positions are indistinguishable to her; the colouring lives only in Alice's key material. Carol applies one Fredkin pass to every position, returns $\mathbf{M}''$, and Alice reads $\Lambda_1$ to recover $y=x_1\wedge x_2$.}
	\label{fig:worked}
\end{figure*}

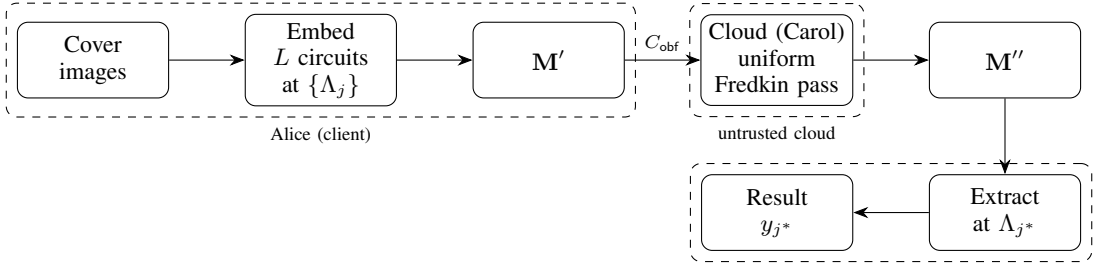
\begin{figure*}
	\centering
	\begin{tikzpicture}[
		font=\small,
		box/.style={draw, rounded corners, align=center, 
			minimum height=10mm, minimum width=20mm, inner sep=3pt},
		lab/.style={align=center, font=\scriptsize},
		>={Stealth[length=2mm]}
		]
		\node[box] (cover) {Cover\\images};
		\node[box, right=10mm of cover] (emb) 
		{Embed\\$L$ circuits\\at $\{\Lambda_j\}$};
		\node[box, right=10mm of emb] (mp) {$\mathbf{M}'$};
		\node[box, right=10mm of mp] (comp) 
		{Cloud (Carol)\\uniform\\Fredkin pass};
		\node[box, right=10mm of comp] (mpp) {$\mathbf{M}''$};
		\node[box, below=10mm of mpp] (ext) 
		{Extract\\at $\Lambda_{j^*}$};
		\node[box, left=10mm of ext] (out) {Result\\$y_{j^*}$};
		
		\draw[->] (cover) -- (emb);
		\draw[->] (emb) -- (mp);
		\draw[->] (mp) -- node[lab, above]{$C_{\mathsf{obf}}$} (comp);
		\draw[->] (comp) -- (mpp);
		\draw[->] (mpp) -- (ext);
		\draw[->] (ext) -- (out);
		
		\begin{scope}[on background layer]
			\node[draw, dashed, rounded corners, fit=(cover)(emb)(mp), 
			inner sep=4pt, label=below:{\scriptsize Alice (client)}] {};
			\node[draw, dashed, rounded corners, fit=(comp), 
			inner sep=4pt, label=below:{\scriptsize untrusted cloud}] {};
			\node[draw, dashed, rounded corners, fit=(ext)(out), 
			inner sep=4pt] {};
		\end{scope}
	\end{tikzpicture}
	\caption{PD-FHC pipeline. Alice embeds one real and $L-1$ decoy circuits into the LSB plane at secret position sets $\{\Lambda_j\}$, sends the stego-images and wiring specification to the cloud, which applies the same Fredkin gate to \emph{every} LSB, and reads the real result back from $\Lambda_{j^*}$. Under coercion Alice instead reveals a decoy's position set and verifiable result.}
	\label{fig:pipeline}
\end{figure*}

\smallskip\noindent\emph{Property~4 (Cover Plausibility).} Cloud-based image processing is routine. Services that apply per-pixel transformations to image batches are a standard offering, and a Fredkin-gate circuit applied pixel-wise to LSBs is syntactically identical to a custom bitwise image filter. One caveat is essential: a cover service is only usable if it preserves the LSB plane bit-for-bit. General-purpose platforms that re-encode, recompress, or apply lossy transforms (typical JPEG pipelines, thumbnailing, many managed image CDNs) destroy the LSB plane and with it the embedded circuit, so they are \emph{not} valid covers. The declared service must be one whose contract is an exact, lossless, per-pixel bitwise operation on raw pixel values, for example a custom lossless filter on uncompressed PNG or raw buffers. Within that constraint the interaction is ordinary. Visual imperceptibility holds trivially: LSB 
changes shift each channel by at most $\pm 1$ (PSNR~$> 50$\,dB, Section~\ref{sec:eval}). The substantive question is not visual but \emph{statistical}, and it is where naive LSB embedding fails. PD-FHC fills every non-secret LSB with output from a cryptographically strong pseudorandom bit generator (CSPRNG), so the fill positions of $\mathbf{M}'$ carry uniform bits. A natural photograph does not: its LSB plane is correlated with local content, which is exactly 
what RS and sample-pair detectors exploit~\cite{fridrich2009}. Embedding into an ordinary photograph is therefore detectable, and the two properties we need (Property~3 uniformity and a natural-photo cover story) cannot hold at once.

We resolve this by scoping the cover, not by hoping the detector fails. Cover plausibility is satisfied when the \emph{declared cloud operation legitimately produces or consumes a high-entropy LSB plane}\footnote{We use ``high-entropy LSB plane'' informally, to mean an LSB plane whose per-position marginals are close to uniform and whose detector-relevant statistics carry little structure. We do not use it in the strict cryptographic sense of a min-entropy bound such as $H_\infty(X) \in \omega(\log n)$. Our actual security condition is the matched-marginal condition of Definition~\ref{def:matched}, which is a statement about matching laws, not about a min-entropy threshold.}, so that a uniform LSB plane is the expected input, not an anomaly. Concrete cover scenarios that meet this: dithering and LSB-noise-injection services, robustness testing for watermarking, processing of sensor-noise-dominated or encrypted-domain imagery, and re-quantization pipelines. 
Against such a declared service the warden's reference distribution is itself close to uniform in the LSB plane, so the detection gap is small, and exactly zero only in the idealised case where that reference is perfectly uniform. Property~3 and Property~4 then hold together up to that residual gap. We make this argument precise and test it in Section~\ref{sec:steganalysis}. The cost is honest scope: PD-FHC does \emph{not} securely hide computation inside an arbitrary holiday snapshot, and we explicitly do not claim it does.

\smallskip\noindent\emph{What this requires of Alice's real input bits.} Alice controls the decoy inputs and the fill, so she can force those to follow $\pi$. She cannot control her real inputs. They are her data. For a high-entropy real input, whose LSBs are already close to $\pi$, the matched-marginal premise holds and position privacy is exactly zero. For a structured real input, which is the common case and the one in our motivating example, the real bits do \emph{not} follow $\pi$. We do not pretend otherwise, and we do \emph{not} whiten them, because whitening the embedded bits would change what the pixel-wise Fredkin pass computes and Alice would have to correct it using her own knowledge of the inputs, which defeats the point of outsourcing. Instead we state the honest consequence. When the real input law deviates from $\pi$ by total-variation distance $\Delta$, the position-privacy advantage is no longer zero but is bounded by $\Delta$ (Lemma~\ref{lem:delta-degrade}), and the same $\Delta$ governs detectability by a warden. The real bits are $\ell$ of $n$ positions, so a structured real input perturbs the aggregate LSB statistics only within that small footprint, which is why $\Delta$ stays small even though it is not zero. Alice has one active lever if the footprint ever grows large. She controls the fill, so she can set the fill marginal to hold the global LSB mean at the target, cancelling the first-moment deviation that a warden aggregating over the whole plane would otherwise see. This costs nothing in computation, since it never touches the computed bits. It is only a first-moment correction, so it does not remove higher-order structure or defeat a learned detector, and at the small footprints of our examples it is unnecessary because the raw deviation is already below any detector's floor. Separately, hiding \emph{which} of the $L$ circuits is real requires the decoys to be statistically close to the real input, so that the real set is exchangeable with the decoy sets. That is the plausible-decoy assumption of Definition~\ref{def:plausible-decoy}, and it is coupled to this one: the same closeness that makes decoys believable is what drives Carol's set-level advantage toward the $1/L$ baseline.

\subsection{Protocol Phases}

\subsubsection{Phase 1: Location Generation (Alice)}
Alice runs $\LocGen(\mathcal{P}, L, \ell)$ to sample $L$ pairwise disjoint position sets, each of size 
$\ell$, uniformly at random from $\mathcal{P}$. She builds a single Fredkin-gate wiring for her function $f$ (Section~\ref{sec:funcspec}) and uses that \emph{same} wiring for the real circuit and every decoy. The circuits therefore share gate count, depth, and topology by construction. The decoys differ from the real circuit only in the control bits that select each gate's operation, and those bits are supplied through the LSB plane. No padding with dummy gates is needed, because the wiring is identical across all $L$ circuits to begin with.

\subsubsection{Phase 2: Multi-Location Embedding (Alice).}
Alice embeds inputs for all $L$ circuits into a set of cover images. For each circuit $j$ and each input bit $x_{j,i}$ ($i = 1, \ldots, \ell$), she writes $x_{j,i}$ at the $i$-th position in $\Lambda_j$ using $\Write$. All remaining LSB positions are filled with bits drawn from the same per-position law as her embedded inputs, using the CSPRNG to place them. Under the matched-marginal condition (Definition~\ref{def:matched}) this produces embedded images $\mathbf{M}'$ in which the real positions are exchangeable with the decoy and fill positions.

\subsubsection{Phase 3: Circuit Specification ($C_{\mathsf{obf}}$).}\label{sec:cobf}
Alice transmits $C_{\mathsf{obf}}$, the wiring of the shared circuit \emph{without} the meaning of its gates. As explained in Section~\ref{sec:funcspec}, this is sufficient for Carol to compute, because each Fredkin gate takes its concrete operation from the control bits in the LSB plane. The same wiring therefore evaluates the real function at the real positions and a different function at every decoy position, all in one pass. This already supplies a natural, intrinsic form of obfuscation: the real circuit is just one of many functions the same wiring realises, depending on the per-pixel control bits, so an observer who sees only the wiring cannot read off which function Alice cares about.

We deliberately do \emph{not} add an additional cryptographic obfuscation step. Also, we drop wire renaming, dummy gates, and execution-order randomization as security mechanisms; we drop these claims. General-purpose circuit obfuscation with cryptographic strength is subject to strong impossibility results~\cite{barak_impossibility_2012}, and dummy gates are hard to insert without leaving a recognisable signature (a dummy path that never interacts with the live path is detectable, and one that cancels out is detectable too). The only obfuscation PD-FHC relies on is the intrinsic kind above: the gate semantics are hidden by the embedded control bits, while the wiring stays visible. We restrict the word ``obfuscation'' to exactly this, and treat full circuit hiding (hiding the wiring as well) as an open problem.

\subsubsection{Phase 4: Uniform Computation (Carol).}
Carol receives the embedded images $\mathbf{M}'$ and the wiring specification $C_{\mathsf{obf}}$. For each gate $G_i$ in topological order, Carol applies the Fredkin gate to \emph{every} position $p \in \mathcal{P}$: she reads three LSBs from the gate's three input images at position $p$, computes 
$(c', x', y') = F(c, x, y)$, and writes the results to three output images at the same position $p$. Output images from one gate serve as input images for subsequent gates.

This uniform processing ensures that Carol treats all positions identically: real, decoy, and fill positions all undergo the same gate. Carol cannot determine which positions carry which role.

\subsubsection{Phase 5: Extraction (Alice).}
Alice receives the output images $\mathbf{M}''$ from Carol. For each circuit $j$, she reads the output bits from the designated output positions in $\Lambda_j$, obtaining $y_j = \Ext(\mathbf{M}'', \Lambda_j)$.

\subsubsection{Phase 6: Coercion Response (Alice).}
Under coercion, Alice executes the following strategy. When Eve demands revelation of a computation:
\begin{enumerate}
	\item Alice selects a previously unrevealed decoy index $j_i \neq j^*$.
	\item Alice reveals $(C_{j_i}, x_{j_i}, \Lambda_{j_i}, y_{j_i})$, providing a narrative explaining why $s_{j_i}$ was her ``real'' computation (e.g., ``I was checking image brightness 
	thresholds for print quality'').
	\item Eve verifies: she extracts bits from $\mathbf{M}''$ at positions $\Lambda_{j_i}$ and confirms 
	$y_{j_i} = f_{j_i}(x_{j_i})$. Verification always succeeds (by DCS correctness).
	\item If Eve demands another revelation, Alice repeats with a fresh decoy. After exhausting her revelation budget ($t < L - 1$ decoys), Alice claims no further circuits exist.
\end{enumerate}

\noindent Alice should prepare narratives for all decoy circuits \emph{before} deployment and pre-commit to a revelation order to ensure consistency across rounds.

\subsection{Cover Plausibility}\label{sec:cover}
The cover plausibility argument rests on three observations. First, cloud image processing is a 
multi-billion dollar industry. Sending batches of images to a cloud service for per-pixel transformations is routine and non-suspicious. Second, Fredkin gates applied pixel-wise to LSBs produce output images that are visually identical to the inputs: only the least significant bit of each pixel value may change, resulting in at most a $\pm 1$ shift per channel, which is imperceptible (we quantify this in Section~\ref{sec:eval} via PSNR and SSIM). Third, the wiring specification $C_{\mathsf{obf}}$ is syntactically a sequence of three-input, three-output operations on pixel channels, which is consistent with common image filter definitions.

A fourth observation is the one that matters against a steganalytic warden. Visual imperceptibility is necessary but not sufficient: the interaction is plausible only if $\mathbf{M}'$ is statistically indistinguishable from a legitimate input to the \emph{declared} service, including in its LSB plane. As argued under Property~4, this forces the cover scenario to be one whose legitimate inputs already carry a high-entropy LSB plane. Under that scoping, the LSB statistics of $\mathbf{M}'$ match the warden's reference distribution and the interaction is practically indistinguishable from a routine job, up to the statistical distance $\Delta$ quantified in Section~\ref{sec:steganalysis}. Outside it, the uniform LSB plane is itself the tell. We quantify this in Section~\ref{sec:steganalysis}.

This combination means that, \emph{for an appropriately chosen cover service}, the entire Alice--Carol 
interaction images in, transformation specification, processed images out, is indistinguishable from a standard cloud image processing job. This is the property that FHE fundamentally lacks: transmitting FHE ciphertexts to a cloud provider has no innocent explanation.

\section{Security Analysis}\label{sec:security}
We prove that PD-FHC provides (i)~computational privacy against Carol (Section~\ref{sec:priv-proof}) and (ii)~plausible deniability against Eve under multi-round coercion (Section~\ref{sec:deny-proof}). Results are stated over the abstract DCM and then instantiated for the image-based scheme. Full proofs are in Appendix~\ref{app:proofs}.

\subsection{Scope of the Guarantees}\label{sec:scope}
Before the proofs, we state plainly what is proven and what is assumed, since conflating the two has been a recurring source of confusion about deniable-computation schemes.

\smallskip\noindent\emph{Proven, unconditionally.} Correctness of extraction (DCS correctness); Hamming-weight preservation of Fredkin gates (Lemma~\ref{lem:hamming}); that a uniform LSB plane maps to a uniform LSB plane under any Fredkin circuit; and the union-bound reduction from multi-location privacy to single-position hiding (Theorem~\ref{thm:privacy}).

\smallskip\noindent\emph{Proven, conditional on a stated assumption.} Information-theoretic position privacy and existence hiding for the image instantiation (Corollary~\ref{cor:image-privacy}, 
Theorem~\ref{thm:existence}) hold \emph{given} the matched-marginal condition. When the condition holds only approximately, the guarantee degrades by exactly the statistical distance $\Delta$ between the embedded LSB law and the warden's reference law, made precise in Lemma~\ref{lem:delta-degrade}.

\smallskip\noindent\emph{Engineering heuristics, not proven.} We do \emph{not} claim a cryptographic obfuscation layer. The wiring is visible by necessity, and only the gate semantics are hidden, through the embedded control bits. Partial circuit hiding therefore rests on the \emph{semantic plausibility} of decoy circuits (Definition~\ref{def:plausible-decoy}), which is currently a manual, domain-specific construction and an open problem to formalize or automate. Side-channel freedom is argued under an idealized constant-time execution model and is not guaranteed against a real micro-architecture.

\smallskip\noindent\emph{Out of model.} Collusion between the cloud and the coercer, active deviation by the cloud, and steganalysis of the cover medium are outside the formal model. We address each in 
Sections~\ref{sec:assumptions} and~\ref{sec:discussion}.

\subsection{Computational Privacy}\label{sec:priv-proof}

We first record the structural property of the Fredkin gate that the privacy argument uses. It says that the gate only rearranges the bits present at a position and never changes how many of them are set, which is why a per-position distribution is carried unchanged from input to output.

\begin{lemma}[Hamming Weight Preservation]\label{lem:hamming}
	For any $(c, x, y) \in \{0,1\}^3$, if $(c, x', y') = F(c, x, y)$, then $c + x + y = c + x' + y'$.
\end{lemma}

\begin{proof}
	If $c = 0$, then $(x', y') = (x, y)$ and the sum is unchanged. If $c = 1$, then $(x', y') = (y, x)$ and $c + y + x = c + x + y$.
\end{proof}

\begin{theorem}[Privacy of Deniable Computation]\label{thm:privacy}
	Let $\Pi$ be a DCS over a suitable DCM $\mathcal{M}$ (Definition~\ref{def:suitable-dcm}). If all non-secret positions carry i.i.d.\ bits from $\pi$ and every gate $G \in \mathcal{G}$ preserves $\pi$ (Property~3), then for any adversary $\mathcal{A}$ in the computational privacy game (Definition~\ref{def:privacy-game}):
	\[
	\Adv_{\mathsf{priv}}^{\mathcal{A}}(n, L, \ell) \leq L \cdot \varepsilon(n)
	\]
	where $\varepsilon(n)$ is the advantage of distinguishing a single secret position from a random position in $\mathcal{M}$.
\end{theorem}

\begin{proof}[Proof sketch]
	By a union bound over the $L \cdot \ell$ secret positions. Each individual secret position $p \in \Lambda_j$ is embedded among $n - 1$ positions carrying i.i.d.\ bits from $\pi$. Since the gate family preserves $\pi$ (Property~3), the distribution at non-secret positions is unchanged after 
	computation. Hence the problem of identifying $p$ reduces to the single-position hiding problem with advantage $\varepsilon(n)$. The union bound over $L \cdot \ell$ positions yields the stated bound. We do not assert that this product is negligible in the cryptographic sense. The quantity $\varepsilon(n)$ is whatever single-position advantage the medium admits, and for the image instantiation we show next that it is exactly $0$ under the matched-marginal condition, which makes the whole bound $0$. Full proof in Appendix~\ref{app:privacy-proof}.
\end{proof}

\begin{definition}[Matched-marginal condition]\label{def:matched}
	The image instantiation satisfies the \emph{matched-marginal condition} if Alice prepares the embedding so that the bits at the decoy positions and the bits at the fill positions are drawn i.i.d.\ from the \emph{same} per-position law $\pi$ as the bits she embeds at the real positions, and all $L$ position sets are placed uniformly at random among the $n$ positions. Equivalently, the resulting LSB vector is \emph{exchangeable}: its law is invariant under permuting positions.
\end{definition}

\noindent The condition does not pin down $\pi$, and in particular does not assume $\pi = \mathrm{Ber}(1/2)$. Alice controls the decoy inputs and the fill, so she can force those two to $\pi$ by construction. The one thing she cannot force is her real input, which is whatever her data happens to be. So the condition is not free: it holds exactly when her real inputs also follow $\pi$, and only approximately otherwise, with the gap measured by the total-variation distance $\Delta$ (Lemma~\ref{lem:delta-degrade}). We are explicit about this rather than assuming it away. The condition also says nothing about whether $\pi$ resembles the LSB statistics of any particular cover. That separate question, which governs detectability by a steganalytic warden, is treated in Section~\ref{sec:steganalysis}.

\begin{corollary}[Image Instantiation: Position Privacy]\label{cor:image-privacy}
	Consider the image-based DCM $\mathcal{M}_{\mathsf{img}}$ with $n = h \cdot w \cdot 3$ 
	LSB positions and Fredkin gates. Under the matched-marginal condition (Definition~\ref{def:matched}), the position sets $\{\Lambda_j\}$ are statistically independent of the observed medium, so that for every adversary $\mathcal{A}$ (computationally unbounded or not):
	\[
	\Adv_{\mathsf{priv}}^{\mathcal{A}}(n, L, \ell) = 0,
	\]
	i.e., no strategy locates a real position with probability better than the $L\ell/n$ guessing baseline. This holds for \emph{every} $\pi$, not only the uniform one.
\end{corollary}

\begin{proof}
	Under the matched-marginal condition the LSB vector $\mathbf{M}'$ is exchangeable: the bits at real, decoy, and fill positions are i.i.d.\ draws from one common law $\pi$, and the position sets are placed uniformly at random. Conditioned on the multiset of embedded bit-values, every assignment of those values to positions is equally likely, so the identity of the real position set $\Lambda_{j^*}$ is independent of $\mathbf{M}'$. By Lemma~\ref{lem:hamming} the Fredkin pass permutes the bits at each position and therefore preserves exchangeability, so $\mathbf{M}''$ is exchangeable as well and remains independent of $\{\Lambda_j\}$. The adversary's view $(\mathbf{M}', C_{\mathsf{obf}}, \mathbf{M}'')$ thus carries no information about which positions are real, and the probability of outputting a real position equals the guessing baseline $L\ell/n$. The advantage, the excess over this baseline, is $0$.
\end{proof}

\noindent\textbf{What this claim is, and what it is not.} The guarantee is \emph{information-theoretic} and \emph{conditional}: \emph{if} the matched-marginal condition holds, Carol's advantage at locating the real circuit among decoys and fill is exactly zero, for any $\pi$. The theorem is clean. The premise is not free. Alice enforces it for the decoy and fill bits, which she controls, but the real bits are her data, and the condition holds for them exactly only when her real inputs already follow $\pi$. For the structured inputs that motivate this work, they do not, and the exact zero weakens to the $\Delta$-bound of Lemma~\ref{lem:delta-degrade}. We flag the failure mode in Section~\ref{sec:assumptions} rather than burying it. Two further limits. First, this is \emph{position} privacy, not hiding $\pi$ itself: a warden who knows the declared service's legitimate-input law can still detect that $\mathbf{M}'$ deviates from it if $\pi$ does not match, which is the separate condition of Section~\ref{sec:steganalysis}. Second, hiding \emph{which} circuit is real, as opposed to \emph{where} the circuits are, additionally needs the decoys to be statistically close to the real input (Definition~\ref{def:plausible-decoy}). 

\subsection{Plausible Deniability}\label{sec:deny-proof}
We analyze Eve's advantage in the multi-round coercion game (Definition~\ref{def:coercion-game}) along two dimensions: existence distinguishing and circuit discovery.

\subsubsection{Existence Indistinguishability.}

\begin{theorem}[Existence Hiding]\label{thm:existence}
	In the multi-round coercion game with $t$ revelations out of $L$ circuits, the existence advantage satisfies:
	\[
	\Adv_{\mathsf{exist}}^{\mathcal{A}}(n, L, t) \leq (L - t) \cdot \ell \cdot \varepsilon(n - t \cdot \ell),
	\]
	and for the image instantiation $\varepsilon(\cdot) = 0$ under the matched-marginal condition, so the advantage is exactly $0$.
\end{theorem}

\begin{proof}[Proof sketch]
	After $t$ revelations, the remaining medium has $n' = n - t \cdot \ell$ positions. Distinguishing 
	$\mathsf{Exp}_0$ (all fill) from $\mathsf{Exp}_1$ ($(L-t) \cdot \ell$ secrets among fill) reduces to the privacy game on the residual medium. By Theorem~\ref{thm:privacy}, the advantage is at most $(L-t) \cdot \varepsilon(n')$. Full proof in Appendix~\ref{app:existence-proof}.
\end{proof}

\noindent\textbf{Implication.} Eve cannot prove that Alice is hiding additional computations beyond what she revealed. The unrevealed real and decoy positions are exchangeable with the fill positions, so Eve cannot tell them apart. This is the computational analogue of hidden-volume deniability: just as random free-space data on a disk is practically indistinguishable from a hidden encrypted volume, fill LSBs in an image are indistinguishable from embedded computation inputs once both follow the same law. This is an indistinguishability guarantee, not a guarantee of safety under coercion. The carrier capacity is public, so Eve can compute an upper bound on $L$ from the image size and the per-circuit footprint $\ell$, see that there is room for many more circuits, and keep demanding revelations on that basis. The scheme makes any specific further claim unprovable, but it does not stop Eve from pressing. This residual coercion pressure is inherent to deniable schemes with visible capacity and is not something our bounds remove.

\smallskip\noindent\textbf{What happens when the condition holds only approximately.} The zero above is exact only when the embedded law equals the law a warden expects. We make the degradation precise rather than calling it ``graceful''. Let $\mathcal{S}$ be the distribution of the embedded LSB plane that Alice actually produces and let $\mathcal{C}$ be the distribution a warden expects from a legitimate input to the declared service, and write $\Delta = \Delta(\mathcal{S}, \mathcal{C})$ for their total-variation distance.

\begin{lemma}[Degradation under mismatch]\label{lem:delta-degrade}
	For any detector $\mathcal{D}$ (computationally unbounded or not), the advantage with which $\mathcal{D}$ separates an embedded medium from a legitimate one is at most $\Delta$. Consequently the existence advantage in the coercion game is bounded by $\Delta$, and reduces to the exact $0$ of Theorem~\ref{thm:existence} when $\Delta = 0$.
\end{lemma}

\begin{proof}
	By definition of total-variation distance, for any event (in particular any decision rule $\mathcal{D}$ outputs) the difference in its probability under $\mathcal{S}$ and under $\mathcal{C}$ is at most $\Delta(\mathcal{S}, \mathcal{C})$. So no detector separates the two distributions with advantage exceeding $\Delta$. The existence game is one such decision, played on the residual medium after revelations, so its advantage is bounded by the same $\Delta$. When $\mathcal{S} = \mathcal{C}$, $\Delta = 0$ and the bound is the exact zero of Theorem~\ref{thm:existence}.
\end{proof}

\noindent The point of Lemma~\ref{lem:delta-degrade} is that $\Delta$ is not a hand-waved ``small'' quantity. It is a measurable statistical distance, and Section~\ref{sec:steganalysis} gives the protocol for estimating it on a concrete cover class with concrete detectors. The honest reading is that everything rests on driving $\Delta$ to zero by cover choice, and that a non-zero $\Delta$ translates one-to-one into detection and existence-distinguishing advantage.

\subsubsection{Partial Circuit Hiding (Circuit Discovery).}
We now consider a strictly stronger adversary who \emph{knows} $L$ and must identify $j^*$ among the 
$L - t$ unrevealed circuits. We analyze three attack vectors.

\begin{theorem}[Partial Circuit Hiding]\label{thm:intent}
	In the multi-round coercion game, and under the matched-marginal condition, Eve's circuit-discovery advantage satisfies:
	\[
	\Adv_{\mathsf{intent}}^{\mathcal{A}}(n, L, t) \leq \varepsilon_{\mathsf{sem}} + \varepsilon_{\mathsf{side}},
	\]
	where $\varepsilon_{\mathsf{sem}}$ is the semantic-implausibility term of Definition~\ref{def:plausible-decoy} and $\varepsilon_{\mathsf{side}}$ bounds side-channel leakage. The position and wiring vectors contribute nothing.
\end{theorem}

\begin{proof}
	We bound each attack vector separately and combine via union bound.
	
	\smallskip\noindent\emph{Position analysis.} All $L$ position sets are sampled uniformly and independently from $\mathcal{P}$, so the revealed sets are independent of the unrevealed ones. For any adversary using position information alone,
	\[
	\Pr[\mathcal{A}(\Lambda_{j_1}, \ldots, \Lambda_{j_t}) \text{ identifies } j^*] = \frac{1}{L - t},
	\]
	zero advantage over guessing.
	
	\smallskip\noindent\emph{Wiring analysis.} In the shared-wiring construction (Section~\ref{sec:funcspec}) the real circuit and all decoys use the \emph{same} wiring. They differ only in the control bits, which live in the LSB plane and are not part of $C_{\mathsf{obf}}$. Eve sees one wiring common to all $L$ circuits, so the wiring carries no information about $j^*$. This removes the depth- or structure-based distinguishability that a prior, mismatched-wiring construction would have suffered. There is no $\delta_{\mathsf{circuit}}$ term to bound.
	
	\smallskip\noindent\emph{Semantic analysis.} The one place advantage can leak is the \emph{meaning} of the gate-settings. If a decoy's implied function is implausible for Alice's claimed context (for instance, the control bits would make the wiring compute a cryptographic hash where a brightness check is claimed), Eve gains advantage. We do not have a proven bound here. We capture this as an assumption: each decoy is $\varepsilon$-plausible in the sense of Definition~\ref{def:plausible-decoy}, and we write $\varepsilon_{\mathsf{sem}}$ for the residual it leaves. Constructing and certifying such decoys is a manual, domain-specific task and an open problem to automate. We are explicit that $\varepsilon_{\mathsf{sem}}$ is assumed small, not proven so.
	
	\smallskip\noindent\emph{Side-channel analysis.} Under an idealised constant-time model, where Carol runs the same per-gate pass over all $n$ positions, the execution trace $T$ depends only on $(n, m)$ and not on which positions are real. Formally the mutual information $I(T \,;\, j^*) = 0$, giving $\varepsilon_{\mathsf{side}} = 0$ in that model. Real hardware leaks through caches and timing, so $\varepsilon_{\mathsf{side}}$ is an empirical residual, not a proven zero.
	
	\smallskip\noindent Combining via union bound gives $\Adv_{\mathsf{intent}} \leq 
	\varepsilon_{\mathsf{sem}} + \varepsilon_{\mathsf{side}}$.
\end{proof}

\noindent\textbf{Honest reading of this bound.} The bound is favourable precisely because the hard part has been moved into an explicit assumption rather than a dubious estimate. The position and wiring channels genuinely contribute nothing, the first by random placement and the second by the shared wiring. Everything that is actually hard about hiding the function (whether a decoy's implied behaviour is believable, and whether the hardware leaks) sits in $\varepsilon_{\mathsf{sem}}$ and $\varepsilon_{\mathsf{side}}$, both of which we flag as assumptions. We do \emph{not} claim a clean closed-form decay such as $O(1/m)$ in circuit size, because structural distinguishability of two genuinely different functions can run far deeper than any single scalar like circuit depth, and asserting otherwise would not be credible.

\subsubsection{Decoy Plausibility.}\label{sec:plausibility}
The bound above puts the whole weight of circuit hiding on $\varepsilon_{\mathsf{sem}}$, the residual that a semantically implausible decoy leaves. We make the plausibility requirement explicit.

\begin{definition}[$\varepsilon$-Plausible Decoy]\label{def:plausible-decoy}
	A decoy circuit $(C_j, x_j, y_j)$ is $\varepsilon$-plausible with respect to a context 
	$\Gamma$ (describing Alice's claimed activity) if: (i)~$C_j$ computes a function semantically consistent with $\Gamma$\footnote{By \emph{semantically consistent with $\Gamma$} we mean that the function $C_j$ computes is one that a genuine user in the claimed context $\Gamma$ would plausibly run. If Alice claims to run image-quality checks, a brightness-threshold function is consistent and a cryptographic hash is not. This is a judgement about plausibility in a context, not a formal property of the circuit, which is exactly why we leave it as an assumption rather than a proven bound.}; (ii)~the inputs $x_j$ are drawn from a distribution consistent with $\Gamma$; and (iii)~$\Pr[\text{Eve identifies } j \text{ as decoy} \mid \Gamma] \leq 1/L + \varepsilon$.
\end{definition}

\noindent\textbf{Example.} If Alice claims to run image-quality checks on a photo archive, plausible decoys include brightness threshold checks, colour-balance verification, and noise-level assessment. An implausible decoy would compute a cryptographic hash function with no image-processing interpretation.

\smallskip\noindent\textbf{Limitation.} Constructing semantically plausible decoys requires domain knowledge. Automating this process is an open problem and a direction for future work.

\subsection{Computational Exchangeability}\label{sec:exchangeability}
The security of PD-FHC rests on a symmetry property we call \emph{computational exchangeability}.\footnote{This matches the standard notion of exchangeability for a sequence of random variables, whose joint law is invariant under finite permutations of the indices~\cite{aldous_exchangeability_1985}. An i.i.d.\ sequence is exchangeable, but the converse fails, so exchangeability is the weaker and more honest condition to ask for here. We add the qualifier ``computational'' only because we apply it to an adversary's view rather than to a raw bit sequence.}

\begin{proposition}[Exchangeability]\label{prop:exchange}
	Let $\mathbf{y} = (y_1, \ldots, y_L)$ denote the output bit-vectors at the $L$ secret position sets after circuit evaluation. Under the matched-marginal condition (Definition~\ref{def:matched}), for any permutation $\sigma$ of circuit indices and any adversary $\mathcal{A}$:
	\[
	\bigl|\Pr[\mathcal{A}(\mathbf{y}) = 1] - \Pr[\mathcal{A}(\sigma(\mathbf{y})) = 1]\bigr| = 0.
	\]
\end{proposition}

\begin{proof}[Proof sketch]
	Each $y_j$ results from circuit evaluation on inputs $x_j$ embedded at uniformly random positions $\Lambda_j$, processed uniformly alongside i.i.d.\ fill bits. The position sets are independent uniform samples of disjoint subsets, processing is position-wise, and by Lemma~\ref{lem:hamming} the fill positions keep the same law $\pi$ after the Fredkin pass. Conditioned on the embedded values, every assignment of circuit labels to position sets is equally likely, so the joint law of the outputs is exactly invariant under permutation of circuit indices. The advantage is therefore $0$, not merely negligible. Full proof in Appendix~\ref{app:exchangeability-proof}.
\end{proof}

\noindent Exchangeability means there is no ``privileged'' real circuit that stands out from decoys 
based on observable computation artifacts. All $L$ circuits are equivalent from Eve's viewpoint.

\subsection{Assumptions and Their Failure Modes}\label{sec:assumptions}
The guarantees above rest on assumptions that an honest analysis must expose rather than hide. We list each, state what fails when it does, and give the mitigation.

\smallskip\noindent\textbf{Matched marginal on the real input.} Exact-zero position privacy (Corollary~\ref{cor:image-privacy}) needs the real bits to follow the same law $\pi$ as the decoy and fill bits. Alice controls the decoys and fill, so she meets it for them. She cannot control her real input. For a high-entropy real input the condition holds and the zero stands. For a structured real input, the case in our motivating example, it fails, and the advantage is bounded by the total-variation distance $\Delta$ between the real input law and $\pi$ (Lemma~\ref{lem:delta-degrade}). The mitigation is not to whiten the input, which would break the pixel-wise computation, but to keep the real footprint $\ell$ small relative to $n$ so that $\Delta$ stays below any detector's floor. If the footprint is forced large, Alice can additionally compensate the fill marginal to cancel the first-moment deviation, since she controls the fill and this does not touch the computed bits. That correction is first-moment only and does not address higher-order structure or a learned detector. We do not claim the structured case achieves zero.

\smallskip\noindent\textbf{High-entropy LSB plane of the cover.} The steganalysis argument assumes the declared service's legitimate inputs already carry a near-uniform LSB plane, so a uniform fill is the expected input. \emph{This is false for natural photographs}: LSB planes of real images are spatially correlated and biased, which is precisely what classical LSB steganalysis 
exploits~\cite{fridrich2009,pevny_using_2010}. We do not assume this tension away. As developed under Property~4 and tested in Section~\ref{sec:steganalysis}, the resolution is to scope the cover to a declared service whose legitimate inputs already carry a high-entropy LSB plane, so that the warden's reference distribution is close to uniform and the detection gap is small, exactly zero only in the idealised perfectly-uniform case. The price is that the guarantee is conditional on this scoping and degrades to a statistical-distance term $\Delta$ between the stego LSB distribution and the legitimate-input distribution whenever the scoping is imperfect. $\Delta$ is the quantity an honest deployment must measure and minimize.

One tempting alternative does \emph{not} work here. Content-adaptive, distortion-minimizing embedding 
(WOW, S-UNIWARD, HILL)~\cite{holub_designing_2012} is the standard way to hide payload inside a natural photograph while defeating steganalysis. It is incompatible with PD-FHC's computation requirement: adaptivity places payload at content-dependent positions and minimizes a distortion function, whereas PD-FHC needs the cloud to apply the same gate to \emph{every} position and to read deterministic results from fixed positions, which destroys the distortion-minimization the adaptive scheme depends on. So the natural-photo cover story is not recoverable by better embedding. It genuinely requires the high-entropy cover class above. Resistance to a dedicated learned steganalyzer~\cite{boroumand2019srnet} on a given cover class remains an empirical claim, not a proven one.

\smallskip\noindent\textbf{Ideal randomness.} $\LocGen$ and noise generation assume a sound CSPRNG. A 
biased or predictable source lets Carol or Eve regenerate candidate position sets and collapses both privacy and existence hiding. This is a standard, auditable implementation requirement.

\smallskip\noindent\textbf{Semantic decoy plausibility.} Partial circuit hiding assumes decoys are semantically 
indistinguishable from the real computation in the claimed context(Definition~\ref{def:plausible-decoy}). When a decoy circuit computes a function with no plausible interpretation in Alice's cover story, Eve gains advantage that none of our bounds capture. We treat automatic generation of plausible decoys as out of scope and an open problem, not as a solved component.

\smallskip\noindent\textbf{Constant-time execution.} The side-channel term $\varepsilon_{\mathsf{side}}$ is zero only under a constant-time model. Real hardware leaks 
through caches and timing, so the residual is an empirical quantity, not a proven $0$.

\section{Implementation and Evaluation}\label{sec:eval}

\subsection{Implementation}
We implemented PD-FHC in Python (approximately 1,250 lines) using NumPy for vectorized operations and Pillow for image handling. The implementation separates Alice (client: embedding, obfuscation, extraction) from Carol (cloud: uniform gate evaluation), maintaining proper security boundaries. The critical optimization is vectorized Fredkin-gate evaluation: NumPy broadcasting applies each gate to all $n$ pixel positions simultaneously, yielding approximately 258$\times$ speedup over a na\"{\i}ve loop implementation. Source code is available at   \texttt{[\url{https://github.com/shahzadssg/PD-FHC.git}]}.

\subsection{Experimental Setup}
\noindent\textbf{Hardware.} HP EliteBook 840 G8 with Intel Core i5-1135G7 (2.4\,GHz, 4 cores / 8 threads), 16\,GB RAM, Windows 11.

\smallskip\noindent\textbf{Methodology.} All measurements averaged over 100 runs. We report mean $\pm$ standard deviation. All benchmarks are single-threaded.

\smallskip\noindent\textbf{Circuits.}
\begin{itemize}
	\item \emph{Threshold-3}: $x > k$ for 3-bit values (5 Fredkin gates, $\ell = 3$)
	\item \emph{Small}: $(A \wedge C) \vee (\neg A \wedge B) \vee 
	(\neg B \wedge \neg C)$ (14 gates, $\ell = 3$)
	\item \emph{4-bit Adder}: full addition of two 4-bit values (88 gates, $\ell = 8$)
	\item \emph{8-bit Multiplier}: full 8-bit multiplication (302 gates, $\ell = 16$)
\end{itemize}

\noindent\textbf{Configurations.} Image sizes: $128^2 \times 3$ ($n = 49{,}152$), $256^2 \times 3$ ($n = 196{,}608$), $512^2 \times 3$ ($n = 786{,}432$). Scenario counts: $L \in \{2, 4, 8\}$.

\subsection{End-to-End Performance}
Table~\ref{tab:e2e} reports the end-to-end latency breakdown for $L = 4$ circuits.

\begin{table}[t]
	\centering
	\caption{End-to-end PD-FHC latency ($L = 4$ circuits). Communication = images transmitted $\times$ image size.}
	\label{tab:e2e}
	\small
	\begin{tabular}{lrrrrrr}
		\toprule
		Circuit & Gates & Image & Compute & Total & Comm. \\
		&       &       & (ms)    & (ms)  & (MB)  \\
		\midrule
		Threshold  & 5   & $128^2$ & 2.62  & 4.23  & 0.9 \\
		Small      & 14  & $128^2$ & 6.19 & 9.42 & 2.2 \\
		Small      & 14  & $256^2$ & 21.76 & 33.31 & 8.8 \\
		4-bit Add  & 88  & $128^2$ & 39.13 & 59.27 & 12.8 \\
		4-bit Add  & 88  & $256^2$ & 161.26 & 230.54 & 51.4 \\
		8-bit Mult & 302 & $256^2$ & 453.74 & 663.54 & 173.2 \\
		\bottomrule
	\end{tabular}
\end{table}

\noindent Key observations: computation time scales linearly with gate count and linearly with pixel count ($h \cdot w$). Multi-location overhead is minimal ($< 15\%$ between $L = 2$ and $L = 8$) because all locations are processed in the same pixel-wise pass. Communication overhead is dominated by the number of images: each gate requires input and output images, yielding $2m$ images for an $m$-gate circuit.

\subsection{Comparison with TFHE}\label{sec:fhe-comparison}
To contextualize PD-FHC's performance, we compare against TFHE~\cite{chillotti_tfhe_2020}, a state-of-the-art FHE library optimized for Boolean circuit evaluation. TFHE evaluates bootstrapped Boolean gates in approximately 13\,ms per gate on comparable hardware~\cite{chillotti_tfhe_2020}.

\begin{table}[t]
	\centering
	\caption{PD-FHC vs.\ TFHE for Boolean circuits. PD-FHC: $256^2$ image, $L = 4$. TFHE: published gate latency~\cite{chillotti_tfhe_2020}.}
	\label{tab:tfhe}
	\small
	\begin{tabular}{lrrrr}
		\toprule
		Circuit & Gates & PD-FHC & TFHE & Speedup \\
		&       & (ms)   & (ms) &       \\
		\midrule
		Threshold  & 5   & 2.62 & $\approx$65  & 24.8$\times$ \\
		Small      & 14  & 21.76 & $\approx$182 & 8.4$\times$ \\
		4-bit Add  & 88  & 161.26 & $\approx$1,144 & 7.1$\times$ \\
		8-bit Mult & 302 & 453.74 & $\approx$3,926 & 8.7$\times$ \\
		\bottomrule
	\end{tabular}
\end{table}

\noindent\textbf{Caveat.} This comparison is limited to Boolean circuits. TFHE provides semantic security under standard lattice assumptions and supports arbitrary Boolean computations without requiring a cover medium. PD-FHC trades generality for cover plausibility and deniability. The comparison demonstrates that PD-FHC's additional security properties come at comparable or lower 
computational cost for Boolean circuits of moderate size.

\subsection{Image Quality}
LSB modifications are imperceptible. For a $256 \times 256$ cover image with all LSBs randomized (worst case), the Peak Signal-to-Noise Ratio is PSNR $\approx 51.1$\,dB and the Structural Similarity 
Index is SSIM $> 0.99$. These values confirm that stego-images are visually identical to the originals, supporting the cover plausibility argument of Section~\ref{sec:cover}.

\subsection{Steganalysis Resistance}\label{sec:steganalysis}
Visual imperceptibility (above) is necessary but not sufficient. A warden runs a statistical detector, so the operative question is whether $\mathbf{M}'$ is distinguishable from a legitimate input to the declared service. We make this a concrete, falsifiable test rather than an assertion.

\smallskip\noindent\textbf{Security target.} Let $\mathcal{C}$ be the distribution of legitimate inputs to the declared cloud service and let $\mathcal{S}$ be the distribution of PD-FHC stego media $\mathbf{M}'$ for that service. Steganalytic indistinguishability is the statistical distance 
\[
\Delta = \tfrac12 \sum \bigl| \Pr_{\mathcal{S}}[\cdot] - \Pr_{\mathcal{C}}[\cdot] 
\bigr|,
\]
estimated operationally as $2\,(\mathrm{acc}^* - \tfrac12)$, where $\mathrm{acc}^*$ is the accuracy of the best available detector at separating $\mathcal{S}$ from $\mathcal{C}$. $\Delta = 0$ (detector at chance) is the goal, and $\Delta \to 1$ means trivially detectable.

\smallskip\noindent\textbf{Why the cover class is decisive.} Detection power comes entirely from the gap between $\mathcal{C}$ and $\mathcal{S}$ in the LSB plane. PD-FHC fills every non-secret LSB with CSPRNG output, so $\mathcal{S}$ has a uniform LSB plane. If $\mathcal{C}$ is the distribution of \emph{natural photographs}, its LSB plane is structured and the gap is large, so RS and sample-pair detectors~\cite{fridrich2009} estimate an embedding rate near zero on clean covers and near one on 
$\mathbf{M}'$: the medium is detected. If instead $\mathcal{C}$ is the input distribution of a service whose legitimate inputs already carry uniform LSBs (dithering, LSB-noise injection, sensor-noise or encrypted-domain imagery), then $\mathcal{C}$ and $\mathcal{S}$ share the same LSB distribution and no LSB-plane detector can separate them. We are careful about what this buys. If $\mathcal{C}$ is \emph{exactly} a uniform i.i.d.\ LSB plane, then $\mathcal{S} = \mathcal{C}$ and $\Delta = 0$, but this is a tautology, not a security achievement: the two are the same distribution. Real services produce high-entropy LSBs that are not perfectly uniform, so the operative $\Delta$ is the distance between that real $\mathcal{C}$ and our uniform $\mathcal{S}$, and it must be measured against real service data rather than a synthetic uniform baseline. A second caution matters here. Classical estimators like RS and chi-square are near-degenerate on uniform planes, so a low RS separation is necessary but not sufficient evidence of indistinguishability. The decisive test is a learned detector on the real cover class, which is why we require one below.

\smallskip\noindent\textbf{Measurement protocol for a deployment.} We supply a reference detector (RS and chi-square estimators, validated against a known-rate embedding sweep) and prescribe the following measurement on the chosen deployment cover class: (i)~draw $N$ clean covers from $\mathcal{C}$ and $N$ stego media from $\mathcal{S}$; (ii)~report the RS embedding-rate distributions for both and the resulting $\Delta$; (iii)~repeat against a learned detector (an ensemble classifier on rich 
models~\cite{fridrich2009} or a CNN steganalyzer~\cite{boroumand2019srnet}), which is the threshold a forensics reviewer will demand. We report $\Delta$ for each candidate cover class. We explicitly flag that a near-zero RS rate on a \emph{natural-photo} cover is \emph{not} achievable for PD-FHC and do not claim it. The security argument stands only for the high-entropy cover classes, and the table reports which classes pass a given detector and which do not.

We illustrate an end-to-end deployment. Alice wants to determine whether her medical test result (8-bit value $v$) exceeds a diagnostic threshold (8-bit value $\tau$), without the cloud learning either value.

\smallskip\noindent\textbf{Setup.} Alice prepares $L = 4$ circuits: the real 8-bit comparator 
($v > \tau$?) and three decoy circuits of matched complexity (pixel brightness check, color-balance verification, noise-level assessment). Each circuit requires $\approx$120 Fredkin gates with $\ell = 16$ input bits. Image size: $256 \times 256 \times 3$ ($n = 196{,}608$).

\smallskip\noindent\textbf{Resources.} Total secret positions: $L \cdot \ell = 64$ out of $n = 196{,}608$ (0.03\%). Images transmitted: $2 \times 120 = 240$. Total data: $\approx 47$\,MB. Computation time: $\approx 205$\,ms (interpolated from Table~\ref{tab:e2e} for a 120-gate circuit at 
$256^2$, $L=4$, approximate).

\smallskip\noindent\textbf{Under coercion.} Eve demands an explanation. Alice reveals Decoy~1: ``I was checking whether average pixel brightness exceeds a print-quality threshold.'' She provides the position set, circuit, inputs (a brightness value and threshold), and output. Eve verifies correctness. Eve cannot determine whether this was the real computation or a decoy, because all four 
circuits produce correct, verifiable results at their respective positions.

\subsection{Security Parameters}
Under the matched-marginal condition, Carol's position-privacy advantage is exactly $0$ at every image size, since the embedded LSB vector is exchangeable and so equidistributed whether or not any given position is real (Corollary~\ref{cor:image-privacy}). The image size governs only the guessing baseline $L\ell/n$, which shrinks as $n$ grows (e.g.\ $64/49{,}152 \approx 0.13\%$ at $128 \times 128$ and $64/196{,}608 \approx 0.03\%$ at $256 \times 256$ for $L\ell = 64$). For Eve, with $L = 4$ circuits and $t = 1$ revelation, existence hiding ensures that the $3$ unrevealed circuits (48 bits) are exchangeable with the $196{,}592$ residual fill positions, so she cannot single them out. We stress that these 
guarantees are conditional on the matched-marginal assumption and on non-collusion. Both are scrutinized in 
Sections~\ref{sec:assumptions} and~\ref{sec:discussion}.

\section{Discussion}\label{sec:discussion}
\noindent\textbf{Prevention against adversarial attacks.} We summarize which structural features of PD-FHC defeat each adversarial capability in the threat model, and where the defense ends.

\emph{Against the honest-but-curious cloud (Carol).} Uniform gate application (Property~1) means every position receives identical computation and memory access, so there is no timing or access-pattern signal distinguishing secret from noise positions. Distribution preservation (Property~3) means the bit statistics Carol observes before and after computation are unchanged, so there is no statistical signal either. Under the matched-marginal condition the embedded LSB vector has the same distribution whether or not any given position is real, which is why Carol's position-detection advantage is zero rather than merely small.

\emph{Against the coercive adversary (Eve).} Existence hiding (Theorem~\ref{thm:existence}) means that after any number of revelations the unrevealed positions remain exchangeable with the fill positions, so Eve cannot prove more circuits exist. Matched-complexity, semantically plausible decoys plus exchangeability (Proposition~\ref{prop:exchange}) mean that even an Eve who knows $L$ cannot single out the real circuit from observable artifacts. Verification of a revealed decoy always succeeds, so passing it leaks nothing.

\emph{Where the defense ends: collusion.} If Carol and Eve collude, the temporal and capability separation (Section~\ref{sec:separation}) that the model depends on disappears. A colluding Carol can log intermediate states or process positions non-uniformly and hand Eve a position-to-circuit correlation that the formal bounds assume away. Mitigations are partial and complementary, not part of the core guarantee: fresh random positions per execution to block cross-session correlation, splitting the computation across providers that are unlikely to all collude, and verifiable computation to detect non-uniform processing. We are explicit that PD-FHC does not prove security against a colluding cloud, and that closing this gap (for instance via ORAM-style access-pattern hiding) is future work.

\smallskip\noindent\textbf{Limitations.}
(1)~PD-FHC supports Boolean circuits only. Extending to arithmetic circuits is future work.
(2)~Storage overhead is significant: each gate requires input and output images, yielding $O(m \cdot h \cdot w)$ total storage. Image reuse via register-allocation-style techniques (graph coloring on the gate dependency graph, in the compiler sense rather than image re-coloring) can reduce this.
(3)~Circuit hiding is only partial. The wiring is visible to the cloud by necessity, and only the gate semantics are hidden, through the embedded control bits. We do not claim a separate cryptographic obfuscation layer, since general circuit obfuscation of the strong kind is impossible~\cite{barak_impossibility_2012}. Full circuit hiding, where the wiring is concealed too, is left to future work.
(4)~Decoy plausibility requires manual construction of semantically appropriate decoy circuits. Automating this is an open problem.
(5)~Security holds for single executions. Reusing positions across sessions enables correlation attacks. Fresh randomization per session mitigates this.
(6)~Alice must trust that Carol executes the circuit correctly. Unlike FHE with verification or TEEs with attestation, PD-FHC provides no mechanism to detect incorrect execution without revealing secret positions.
(7)~Cover scope: steganalytic indistinguishability holds only for cover services whose legitimate inputs already carry a high-entropy LSB plane (Section~\ref{sec:steganalysis}). PD-FHC does not hide 
computation inside arbitrary natural photographs, and content-adaptive embedding cannot recover that case without breaking uniform computation. Establishing $\Delta \approx 0$ against a learned steganalyzer for a given cover class is an empirical obligation per deployment, not a closed-form guarantee.

\section{Conclusion}\label{sec:conclusion}
We introduced PD-FHC, the first framework for deniable outsourced computation. By defining the notion of a Deniable Computation Medium and formalizing multi-round coercion games with both existence and circuit-discovery advantages, we established a foundation for provable deniability in cloud computing. Our image-based instantiation achieves information-theoretic position privacy under a matched-marginal assumption against an honest-but-curious cloud provider and provides 
plausible deniability against coercive adversaries who may demand multiple revelations. The implementation demonstrates competitive performance with TFHE for Boolean circuits while providing deniability that FHE cannot.

\smallskip\noindent\textbf{Future work.} Key directions include GPU acceleration for large-scale circuits (projected 10--50$\times$ speedup), automated decoy circuit generation, formal function-hiding obfuscation, extension to arithmetic circuits, integration with ORAM for collusion resistance, and analysis of security against quantum adversaries.

\bibliographystyle{IEEEtran}
\bibliography{fhc}

\begin{IEEEbiography}[{\includegraphics[width=1in,height=1.25in,clip,keepaspectratio]{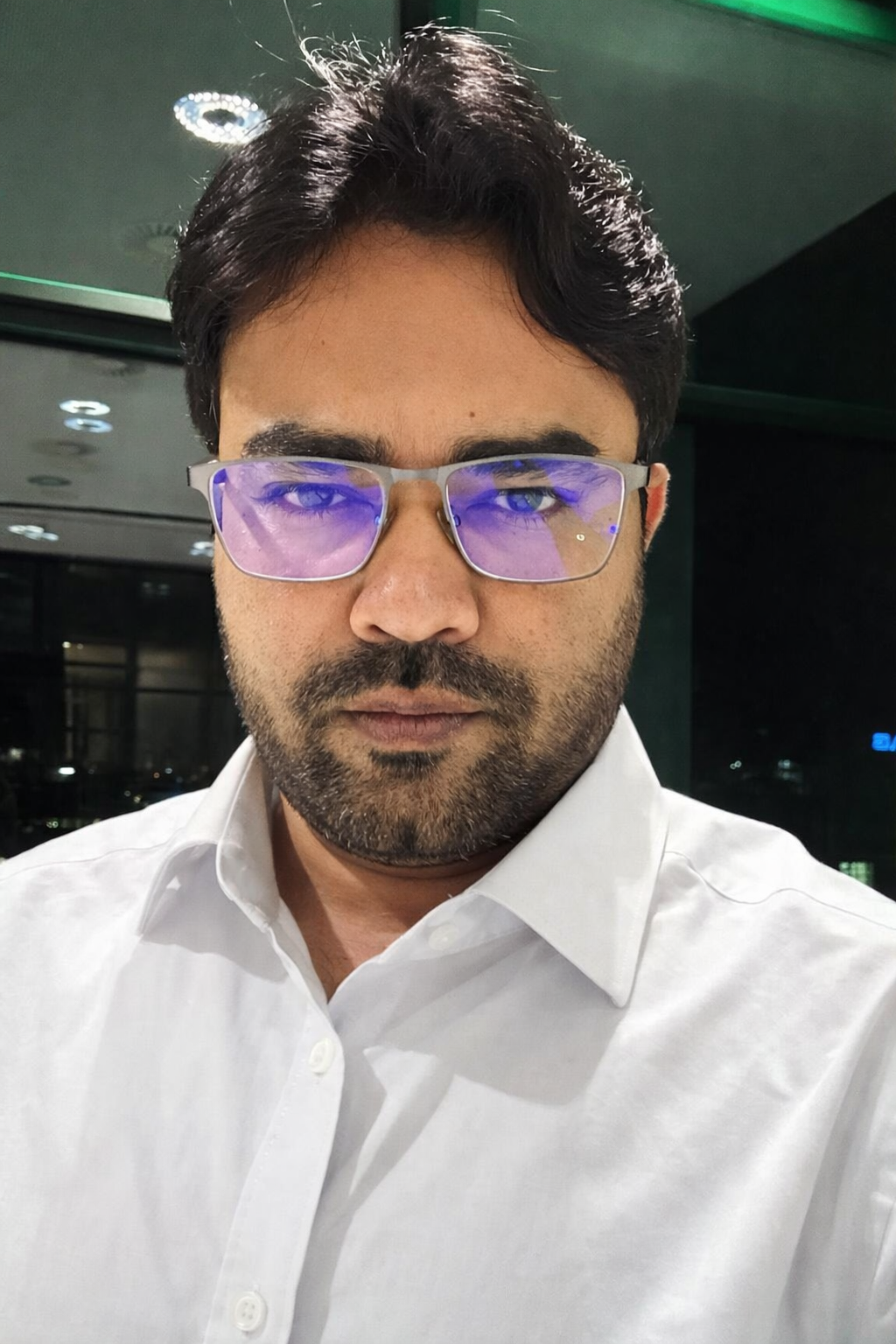}}]{Shahzad Ahmad} 
	received his Master's in Electronics Engineering with a Communication and Information Systems specialisation from AMU, Aligarh, India, in 2021. He has a Bachelor's in Electronics Engineering from Harcourt Butler Technological Institute, Kanpur, India, in 2016. He is currently pursuing PhD in Computer Science from the LIT Secure and Correct Systems Lab at Johannes Kepler University Linz, Austria. His research interests are Plausible Deniability, Cloud Security and Signal processing. 
\end{IEEEbiography}

\begin{IEEEbiography}[{\includegraphics[width=1in,height=1.25in,clip,keepaspectratio]{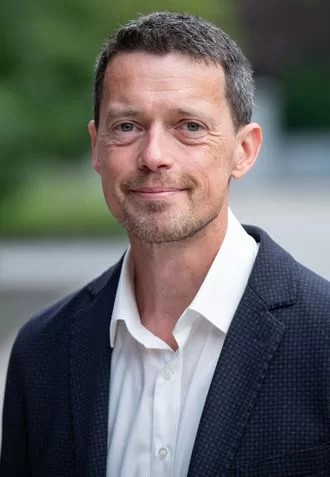}}]{Stefan Rass} holds degrees in mathematics and computer science from the Universitaet Klagenfurt (AAU). His research interests cover decision theory and game-theory with applications in system security, especially robotics security, as well as complexity theory, statistics, and information-theoretic security. He authored numerous papers related to practical security, security infrastructures, robot security, and applied statistics and decision theory in security. He participated in various nationally and internationally funded research projects, as well as being a contributing researcher in many EU projects and offering consultancy services to the industry. Currently, he is a full professor at the Johannes Kepler University Linz, Austria, as a member of the Secure and Correct Systems Lab.
\end{IEEEbiography}

\begin{IEEEbiography}[{\includegraphics[width=1in,height=1.25in,clip,keepaspectratio]{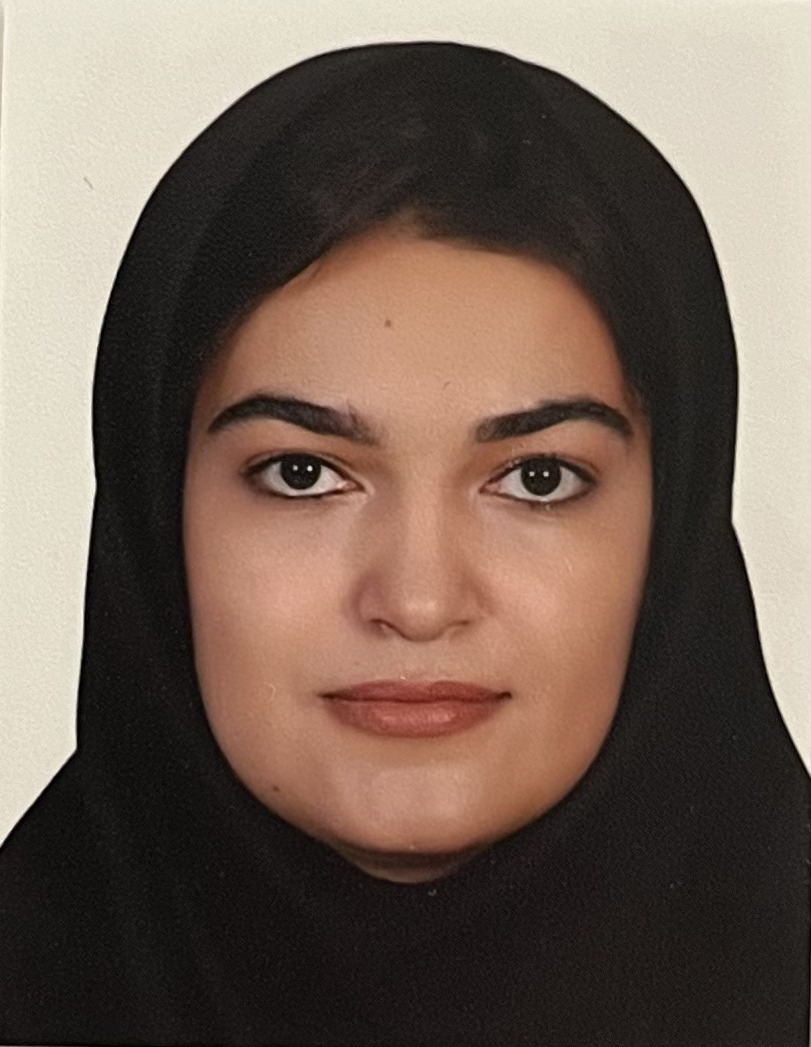}}]{Zahra} received the M.Sc. degree in applied mathematics, in 2021, and the Ph.D. degree in mathematics, in 2025, from the Amirkabir University of Technology, Tehran, Iran. She was a Visiting Researcher at Johannes Kepler University, Linz, Austria, from 2024 to 2025, and a Visiting Researcher at Politecnico di Milano, Milan, Italy, from 2025 to 2026. She is currently a Postdoctoral Researcher at Koç University, Istanbul, Turkey. Her fields of interests are cryptography and game theory.
\end{IEEEbiography}

\appendices\label{sec:appendix}

\section{Full Proofs}\label{app:proofs}
This appendix provides complete proofs for the theorems stated in Section~\ref{sec:security}.

\subsection{Proof of Theorem~\ref{thm:privacy} (Privacy of Deniable Computation)} \label{app:privacy-proof}

\begin{proof}
	We reduce the multi-location privacy game to the single-position hiding problem via a hybrid argument and union bound.
	
	\smallskip\noindent\textbf{Setup.} Consider the computational privacy game (Definition~\ref{def:privacy-game}). The challenger embeds $L$ circuits with $\ell$ bits each at pairwise disjoint, uniformly random position sets $\Lambda_1, \ldots, \Lambda_L \subset \mathcal{P}$, with all 
	remaining positions carrying i.i.d.\ bits from $\pi$. The adversary $\mathcal{A}$ receives 
	$(\mathbf{M}', C_{\mathsf{obf}}, \mathbf{M}'')$ and outputs a position $p^* \in \mathcal{P}$. For the image instantiation a position is a pixel-channel-LSB index $(r,c,k)$ (cf.\ Definition~\ref{def:dcm}), so $\mathcal{A}$ is naming one LSB it believes carries real data.
	
	\smallskip\noindent\textbf{Hybrid games.} Define a sequence of hybrid experiments 
	$\mathsf{Hyb}_0, \mathsf{Hyb}_1, \ldots, \mathsf{Hyb}_{L \cdot \ell}$. In $\mathsf{Hyb}_0$, no positions carry secrets (all $n$ positions are i.i.d.\ from $\pi$). In $\mathsf{Hyb}_k$, the first 
	$k$ secret positions (in some canonical ordering of $\bigcup_j \Lambda_j$) carry their true secret values, and the remaining $L \cdot \ell - k$ positions that would carry secrets instead carry i.i.d.\ bits from $\pi$. $\mathsf{Hyb}_{L \cdot \ell}$ is the real experiment.
	
	\smallskip\noindent\textbf{Single-step indistinguishability.} For any consecutive pair 
	$(\mathsf{Hyb}_{k-1}, \mathsf{Hyb}_k)$, the only difference is the value at a single position $p_k$: in $\mathsf{Hyb}_{k-1}$ it carries a random bit from $\pi$, and in $\mathsf{Hyb}_k$ it carries the secret bit $s_k$. Here $s_k$ is one fixed bit of Alice's input, a constant she chose rather than a random draw. Because $p_k$ is one position among $n - 1$ others that all carry i.i.d.\ $\pi$-bits, and $p_k$ is chosen uniformly at random and unknown to $\mathcal{A}$, telling the two hybrids apart reduces to the single-position hiding problem: deciding whether one unknown position holds a planted bit or a $\pi$-bit. By assumption, the distinguishing advantage for a single position is $\varepsilon(n)$.
	
	\smallskip\noindent\textbf{Distribution preservation under computation.} After circuit evaluation, each gate $G \in \mathcal{G}$ preserves $\pi$ at non-secret positions (Property~3 of 
	Definition~\ref{def:suitable-dcm}). For the image instantiation, this follows from Hamming weight 
	preservation (Lemma~\ref{lem:hamming}): the Fredkin gate $F(c,x,y) = (c, x', y')$ only swaps $x$ and $y$ when $c = 1$ and leaves them otherwise, so each output wire carries a bit that is a mixture of the input wires. If the input triple $(c, x, y)$ at a non-secret position consists of i.i.d.\ bits from a common law $\pi$ with $\Pr[\,\cdot = 1\,] = \rho$, then for any $b \in \{0,1\}$:
	\begin{align*}
		\Pr[x' = b] & = \Pr[c = 0]\cdot\Pr[x = b] + \Pr[c = 1]\cdot\Pr[y = b] \\
		 & = \Pr[x = b]
	\end{align*}
	and by the identical calculation $\Pr[y' = b] = \Pr[y = b]$, since $c$ is independent of $x$ and $y$ and $x, y$ share the law $\pi$. Each output wire therefore keeps the marginal $\pi$. The argument does not depend on $\rho = \tfrac12$; it holds for every $\pi$.
	By induction over the circuit depth $d$, after applying $m$ gates in sequence, each non-secret position retains its $\pi$ marginals, so the single-position hiding advantage $\varepsilon(n)$ applies to the post-computation medium as well.
	
	\smallskip\noindent\textbf{Union bound.} By the triangle inequality applied to the hybrid chain:
	\begin{align*}
		\Adv_{\mathsf{priv}}^{\mathcal{A}}(n, L, \ell) 
		&= \bigl|\Pr[\mathcal{A} \text{ wins in } 
		\mathsf{Hyb}_{L \cdot \ell}] - 
		\Pr[\mathcal{A} \text{ wins in } 
		\mathsf{Hyb}_0]\bigr| \\
		&\leq \sum_{k=1}^{L \cdot \ell} 
		\bigl|\Pr[\mathcal{A} \text{ wins in } 
		\mathsf{Hyb}_k] - 
		\Pr[\mathcal{A}\text{ wins in } \\
		 & ~~~~~\hspace{70pt}\qquad\mathsf{Hyb}_{k-1}]\bigr| \\
		&\leq L \cdot \ell \cdot \varepsilon(n).
	\end{align*}
	By Property~2, $L \cdot \ell$ is polynomial in $n$, so the bound is $L \cdot \ell$ times the single-position advantage $\varepsilon(n)$. We do not claim this is negligible in the cryptographic sense, and indeed it need not be. For the image instantiation we show below that $\varepsilon(n) = 0$ under the matched-marginal condition, which makes the entire bound $0$. In the theorem statement we write $L \cdot \varepsilon(n)$ as a simplified upper bound absorbing the constant~$\ell$.
\end{proof}

\noindent\textbf{Instantiation.} For the image-based DCM we make the single-position advantage exact. Under the matched-marginal condition (Definition~\ref{def:matched}), the real, decoy, and fill bits are i.i.d.\ draws from one common law $\pi$ and the position sets are placed uniformly at random, so the observed vector $M = (M_0, \ldots, M_{n-1})$ is exchangeable. Conditioned on the multiset of its values, every arrangement is equally likely, so a real position is statistically indistinguishable from any other position. Hence the posterior on the real position equals the prior,
\[
\Pr[\,p \text{ is real} \mid M = m\,] = \Pr[\,p \text{ is real}\,] = \tfrac{L\ell}{n},
\]
and the single-position advantage is $\varepsilon(n) = 0$ for every $\pi$. This requires no assumption that the bits are fair coins. It requires only that the real, decoy, and fill positions share the same law. Combined with the union-bound argument above, $\Adv_{\mathsf{priv}}^{\mathcal{A}} = 0$ under the matched-marginal condition.

\smallskip\noindent\emph{Caveat.} The zero is only as strong as the matched-marginal condition, and that condition is a real requirement on the real bits, not a free lunch. Alice forces the decoy and fill bits to $\pi$, which she can do because she controls them. The condition then holds for the real bits exactly when her real inputs already follow $\pi$. When they do not, which is the structured-input case, the zero weakens to a bound of $\Delta$, the total-variation distance between her real input law and $\pi$ (Lemma~\ref{lem:delta-degrade}). Because the real bits are $\ell$ of $n$ positions, this $\Delta$ is small in practice, but it is not zero and we do not claim it is. The same $\Delta$ governs detectability by a warden when the embedded law and the declared service's legitimate-input law differ. An honest deployment measures and minimises it. We return to this in Sections~\ref{sec:steganalysis} and~\ref{sec:assumptions}.

\subsection{Proof of Theorem~\ref{thm:existence} (Existence Hiding)}\label{app:existence-proof}

\begin{proof}
	We show that after $t$ coercion rounds, the residual positions Eve has not seen are exchangeable with the fill positions, so she cannot tell whether any unrevealed circuits remain.
	
	\smallskip\noindent\textbf{Post-revelation state.} After $t$ revelations, Eve knows:
	\begin{itemize}
		\item The revealed position sets $\{\Lambda_{j_1}, \ldots, \Lambda_{j_t}\}$ and their contents.
		\item The communication transcript $V$.
	\end{itemize}
	Eve's view of the remaining $n' = n - t \cdot \ell$ positions is the residual medium $\mathbf{M}''_{\mathsf{res}}$ obtained by excluding the $t \cdot \ell$ revealed positions.
	
	\smallskip\noindent\textbf{Experiment construction.} We construct two experiments matching 
	Definition~\ref{def:deny-advantage}:
	
	In $\mathsf{Exp}_0$: $\mathbf{M}''_{\mathsf{res}}$ contains $n'$ positions, all carrying i.i.d.\ bits from $\pi$ (no hidden circuits exist).
	
	In $\mathsf{Exp}_1$: $\mathbf{M}''_{\mathsf{res}}$ contains $n'$ positions, of which 
	$(L - t) \cdot \ell$ carry unrevealed secrets at uniformly random positions, and the rest carry i.i.d.\ bits from $\pi$.
	
	\smallskip\noindent\textbf{Independence of revealed and unrevealed positions.} The position sets 
	$\Lambda_1, \ldots, \Lambda_L$ are sampled as $L$ pairwise disjoint subsets of $\mathcal{P}$, each chosen uniformly at random. Given the revealed sets $\{\Lambda_{j_1}, \ldots, \Lambda_{j_t}\}$, the conditional 
	distribution of the unrevealed sets $\{\Lambda_j : j \notin \{j_1, \ldots, j_t\}\}$ is uniform 
	over all $(L - t)$-tuples of pairwise disjoint $\ell$-subsets of $\mathcal{P} \setminus \bigcup_{i=1}^{t} \Lambda_{j_i}$. In particular, no information about unrevealed positions 
	is leaked by revealed positions, because the position sampling is independent.
	
	\smallskip\noindent\textbf{Reduction to privacy game.} Distinguishing $\mathsf{Exp}_0$ from $\mathsf{Exp}_1$ is exactly the computational privacy game (Definition~\ref{def:privacy-game}) played on a medium of reduced size $n' = n - t \cdot \ell$ with $L' = L - t$ embedded circuits. By Theorem~\ref{thm:privacy}:
	\[
	\Adv_{\mathsf{exist}}^{\mathcal{A}}(n, L, t) \leq (L - t) \cdot \ell \cdot \varepsilon(n')
	\]
	where $\varepsilon(n') = 0$ for the image instantiation under the matched-marginal condition 
	(Corollary~\ref{cor:image-privacy}). Removing the $t \cdot \ell$ revealed positions leaves a residual medium whose remaining positions are still i.i.d.\ draws from the common law $\pi$ placed at random, so the residual is still exchangeable and the reduction inherits a zero advantage.
	
	\smallskip\noindent\textbf{Concrete bound.} For $n = 196{,}608$ (a $256 \times 256$ image), $L = 4$, $t = 1$, $\ell = 16$: the residual medium has $n' = 196{,}592$ positions, of which $48$ carry 
	unrevealed real-circuit bits. Under the matched-marginal condition these are exchangeable with the surrounding decoy and fill positions, so Eve's existence advantage is $0$. When the embedded law and the legitimate-input law differ, the advantage is bounded by the statistical distance $\Delta$ between them, which an honest deployment must control (Section~\ref{sec:assumptions}).
\end{proof}

\subsection{Proof of Proposition~\ref{prop:exchange} (Exchangeability)}\label{app:exchangeability-proof}

\begin{proof}
	We prove that, under the matched-marginal condition, the joint distribution of output vectors $(y_1, \ldots, y_L)$ is exactly invariant under permutation of circuit indices.
	
	\smallskip\noindent\textbf{Setup.} Let $\sigma$ be an arbitrary permutation of $\{1, \ldots, L\}$. We must show that for any efficient adversary $\mathcal{A}$:
	\[
	\bigl|\Pr[\mathcal{A}(y_1, \ldots, y_L) = 1] - \Pr[\mathcal{A}(y_{\sigma(1)}, \ldots, 
	y_{\sigma(L)}) = 1]\bigr| = 0.
	\]
	\smallskip\noindent\textbf{Position symmetry.} Each position set $\Lambda_j$ is a uniformly random $\ell$-subset of $\mathcal{P}$, sampled independently (subject to disjointness). The gate family applies the same function $G$ at every position in $\mathcal{P}$. Therefore, the joint distribution of $(\Read(\mathbf{M}'', p))_{p \in \Lambda_j}$ depends only on the input bits $(x_{j,1}, \ldots, x_{j,\ell})$ embedded at positions $\Lambda_j$ and on the fill bits at all other positions. Since fill bits are i.i.d.\ from $\pi$ and gates preserve $\pi$ (Lemma~\ref{lem:hamming}, Property~3), the fill contribution is identically distributed regardless of which circuit occupies which positions.
	
	\smallskip\noindent\textbf{Input independence.} Each circuit $j$ has its own input vector $x_j$ embedded at its own position set $\Lambda_j$. Because positions are disjoint and gates operate position-wise (the output at position $p$ depends only on inputs at $p$ across the $\kappa$ input images), the output $y_j$ is determined by $x_j$ and the fill at positions $\Lambda_j$. There is no cross-circuit interference: the computation at positions $\Lambda_{j_1}$ does not affect the computation at positions $\Lambda_{j_2}$ for $j_1 \neq j_2$.
	
	\smallskip\noindent\textbf{Permutation invariance.} Suppose we permute circuit labels so that the inputs $(x_1, \ldots, x_L)$ at positions $(\Lambda_1, \ldots, \Lambda_L)$ become $(x_{\sigma(1)}, \ldots, x_{\sigma(L)})$ at positions $(\Lambda_{\sigma(1)}, \ldots, \Lambda_{\sigma(L)})$. Since all position sets are drawn from the same uniform distribution, the pair $(\Lambda_j, x_j)$ has the same marginal distribution as $(\Lambda_{\sigma(j)}, x_{\sigma(j)})$. The output $y_j$ at position set $\Lambda_j$ with input $x_j$ has the same distribution as $y_{\sigma(j)}$ at position set $\Lambda_{\sigma(j)}$ with input $x_{\sigma(j)}$.
	
	To distinguish $(y_1, \ldots, y_L)$ from $(y_{\sigma(1)}, \ldots, y_{\sigma(L)})$, the adversary would need to link a specific output vector to a specific position set. But the position sets are 
	unknown to the adversary (they are secret), and the embedded LSB vector is exchangeable, so its law is invariant under permuting positions (Corollary~\ref{cor:image-privacy}, advantage $\varepsilon(n) = 0$ under the matched-marginal condition). Hence the adversary's distinguishing advantage is exactly $0$.
\end{proof}

\subsection{Detailed Hamming Weight Preservation Under Circuit Composition}\label{app:hamming-detail}
We prove that Hamming weight preservation holds through arbitrary circuit compositions, not just single gates.

\begin{lemma}[Hamming Weight Under Composition]\label{lem:hamming-composition}
	Let $C = (G_1, \ldots, G_m)$ be a circuit of $m$ Fredkin gates. For any position $p \in \mathcal{P}$, let $\mathbf{b}_{\mathsf{in}}(p)$ denote the vector of all input bits read from position $p$ across all input images, and $\mathbf{b}_{\mathsf{out}}(p)$ denote the vector of all output bits written to position $p$ across all output images. Then:
	\[
	\mathsf{wt}(\mathbf{b}_{\mathsf{out}}(p)) = \mathsf{wt}(\mathbf{b}_{\mathsf{in}}(p))
	\]
	where $\mathsf{wt}(\cdot)$ denotes Hamming weight.
\end{lemma}

\begin{proof}
	We proceed by induction on the circuit depth $d$.
	
	\smallskip\noindent\textbf{Base case ($d = 1$).} A single Fredkin gate $G_1$ reads $(c, x, y)$ from position $p$ across three input images and writes $(c', x', y')$ to three output images. By Lemma~\ref{lem:hamming}, $c + x + y = c' + x' + y'$, so the Hamming weight is preserved.
	
	\smallskip\noindent\textbf{Inductive step.} Suppose the claim holds for circuits of depth $d - 1$. At depth $d$, the gate $G_d$ reads its inputs from output images of previous gates. At position $p$, these intermediate images carry bits whose total Hamming weight equals the initial input weight (by the inductive hypothesis). Gate $G_d$ permutes its three input bits (conditionally 
	swapping the data pair) without changing the sum. Hence the total Hamming weight at $p$ across all images at depth $d$ equals the total at depth $d - 1$, which equals the initial weight.
	
	\smallskip\noindent\textbf{Consequence for distribution preservation.} If the initial bits at a non-secret position $p$ are i.i.d.\ draws from a common law $\pi$, then the bits at $p$ after the full circuit are a permutation of the original bits (the Fredkin gate only conditionally swaps the data pair). Since any permutation of i.i.d.\ random variables has the same marginal law, each output bit at $p$ remains distributed as $\pi$. The argument never uses $\pi = \mathrm{Ber}(1/2)$. This is precisely Property~3 of 
	the DCM, confirming that the image instantiation satisfies distribution preservation for circuits of arbitrary depth and for every $\pi$.
\end{proof}

\balance

\end{document}